\RequirePackage[l2tabu,orthodox]{nag}
\documentclass[a4paper,english]{scrartcl}

\usepackage{tikz}
\usetikzlibrary{decorations.pathreplacing,angles,quotes}
\usepackage{standalone}

\usepackage[utf8]{inputenc}
\usepackage[T1]{fontenc}
\usepackage{lmodern} 

\usepackage{microtype}
\usepackage[within=none,figurename=Fig.,labelsep=period]{caption}

\usepackage{amssymb}
\usepackage{amsmath}
\usepackage{bm}   
\usepackage{mathtools}

\allowdisplaybreaks[1]

\usepackage[inline]{enumitem}
\setlist[itemize]{label=$\circ$}
\setlist[description]{labelindent=\parindent}

\usepackage[amsmath,hyperref,thmmarks]{ntheorem}

\theoremseparator{.}
\newtheorem{theorem}{Theorem}
\newtheorem{lemma}[theorem]{Lemma}

\newtheorem{corollary}[theorem]{Corollary}
\newtheorem{claim}[theorem]{Claim}

\theoremstyle{plain}
\theorembodyfont{\upshape}

\newtheorem{definition}[theorem]{Definition}

\theoremstyle{nonumberplain}
\theoremheaderfont{\itshape}
\theorembodyfont{\upshape}
\theoremsymbol{\ensuremath{_\blacksquare}}
\newtheorem{proof}{Proof}

\newcommand{\cc}[1]{\ensuremath{\mathrm{#1}}}

\newcommand{\op}[1]{\ensuremath{\operatorname{#1}}}

\renewcommand{\P}{\cc{P}}
\newcommand{\NP}{\cc{NP}}
\newcommand{\W}{\cc{W[1]}}
\newcommand{\sharpP}{\cc{\#P}}
\newcommand{\poly}{\op{poly}}
\newcommand{\FP}{\cc{FP}}

\newcommand{\N}{\mathbb{N}}

\newcommand{\Q}{\mathbb{Q}}

\usepackage{xparse}
\DeclareDocumentCommand{\restrict}{O{}}{\mathord{\restriction}_{#1}}

\newcommand{\loopG}{\ensuremath{\mathcal{G}^\circ}}
\newcommand{\graphs}{\ensuremath{\mathcal{G}}}
\newcommand{\subtt}{\ensuremath{\#\mathsf{Sub}(\mathcal{T},\mathcal{T})}}

\newcommand{\homs}{\ensuremath{\mathsf{Hom}}}
\newcommand{\lc}{\ensuremath{\mathsf{LC}}}
\newcommand{\lihoms}{\ensuremath{\mathsf{Li}\text{-}\mathsf{Hom}}}
\newcommand{\lirhoms}{\ensuremath{\mathsf{Li}[r]\text{-}\mathsf{Hom}}}
\newcommand{\embs}{\ensuremath{\mathsf{Emb}}}

\newcommand{\Tau}{\mathrm{T}}

\newcommand{\bs}[1]{ \ensuremath{ \{ #1 \} } }

\newcommand{\tauhoms}{\ensuremath{\homs_\tau}}
\newcommand{\taumins}{\ensuremath{\tau\text{-}\mathcal{M}}}
\newcommand{\Taumins}{\ensuremath{\Tau\text{-}\mathcal{M}}}
\newcommand{\tauli}{\ensuremath{\tau_{\mathsf{Li}}}}
\newcommand{\taulihoms}{{\homs_{\tauli}}}
\newcommand{\taulimins}{\ensuremath{\tauli\text{-}\mathcal{M}}}
\newcommand{\taulir}{\ensuremath{\tau_{\mathsf{Li}[r]}}}

\newcommand{\taulirmins}{\ensuremath{\taulir\text{-}\mathcal{M}}}

\title{Counting Restricted Homomorphisms via Möbius Inversion over Matroid Lattices}

\author{Marc Roth \\
Saarbr\"ucken Graduate School of Computer Science \\
Cluster of Excellence (MMCI), Saarland University}
\date{}

\begin{document}
 
\maketitle

\begin{abstract}
We present a framework for the complexity classification of parameterized counting problems that can be formulated as the summation over the numbers of homomorphisms from small pattern graphs $H_1,\dots,H_\ell$ to a big host graph $G$ with the restriction that the coefficients correspond to evaluations of the Möbius function over the lattice of a graphic matroid. This generalizes the idea of Curticapean, Dell and Marx [STOC 17] who used a result of Lov{\'a}sz stating that the number of subgraph embeddings from a graph $H$ to a graph $G$ can be expressed as such a sum over the lattice of partitions of $H$.\\ In the first step we introduce what we call graphically restricted homomorphisms that, inter alia, generalize subgraph embeddings as well as locally injective homomorphisms. We provide a complete parameterized complexity dichotomy for counting
such homomorphisms, that is, we identify classes of patterns for which
the problem is fixed-parameter tractable (FPT), including an algorithm, and prove that all other pattern classes lead to $\#\W$-hard problems. The main ingredients of the proof are the complexity classification of linear combinations of homomorphisms due to Curticapean, Dell and Marx [STOC 17] as well as a corollary of Rota's NBC Theorem which states that the sign of the Möbius function over a geometric lattice only depends on the rank of its arguments.\\
We apply the general theorem to the problem of counting locally injective homomorphisms from small pattern graphs to big host graphs yielding a concrete dichotomy criterion. It turns out that --- in contrast to subgraph embeddings --- counting locally injective homomorphisms has ``real'' FPT cases, that is, cases that are fixed-parameter tractable but not polynomial time solvable under standard complexity assumptions. To prove this we show in an intermediate step that the subgraph counting problem remains $\#\P$-hard when both the pattern and the host graphs are restricted to be trees. We then investigate the more general problem of counting homomorphisms that are injective in the $r$-neighborhood of every vertex. As those are graphically restricted as well, they can also easily be classified via the general theorem.\\
Finally we show that the dichotomy for counting graphically restricted homomorphisms readily extends to so-called linear combinations.
\end{abstract}

\section{Introduction}
In his seminal work about the complexity of computing the permanent Valiant \cite{Valiant1979a} introduced counting complexity which has since then evolved into a well-studied subfield of computational complexity. Despite some surprising positive results like polynomial time algorithms for counting perfect matchings in planar graphs by the FKT method \cite{fkt1,fkt2}, counting spanning trees by Kirchhoff's Matrix Tree Theorem or counting Eulerian cycles in directed graphs using the ``BEST''-Theorem \cite{best}, most of the interesting problems turned out to be intractable. Therefore, several relaxations such as restrictions of input classes \cite{rest1} and approximate counting \cite{appr1,appr2} were introduced. Another possible relaxation, the one this work deals with, is to consider parameterized counting problems as introduced by Flum and Grohe \cite{flumgrohe_counting}. Here, problems come with an additional parameter $k$ and a problem is fixed-parameter tractable (FPT) if it can be solved in time $g(k)\cdot \poly(n)$ where $n$ is the input size and $g$ is a computable function, which yields fast algorithms for large instances with small parameters. On the other hand, a problem is considered intractable if it is $\#\W$-hard. This stems from the fact that $\#\W$-hard problems do not allow an FPT algorithm unless standard assumptions such as the exponential time hypothesis (ETH) are wrong. \\
When investigating a family of related (counting) problems one could aim to simultaneously solve the complexity of as many problems as possible, rather than tackling a (possibly infinite) number of problems by hand. For example, instead of proving that counting paths in a graph is hard, then proving that counting cycles is hard and then proving that counting stars is easy, one should, if possible, find a criterion that allows a classification of those problems in hard and easy cases. Unfortunately, there are results like Ladner's Theorem \cite{ladner}, stating that there are problems neither in $\P$ nor $\NP$-hard (assuming $\P \neq \NP$), which give a negative answer to that goal in general. However, there are families of problems that have enough structure to allow so-called dichotomy results. One famous example, and to the best of the authors knowledge this was the first such result, is Schaefer's dichotomy \cite{schaefer}, stating that every instance of the generalized satisfiability problem is either polynomial time solvable or $\NP$-complete. Since then much work has been done to generalize this result, culminating in recent announcements (\cite{bulatov},\cite{zhuk},\cite{feder}) of a proof of the Feder-Vardi-Conjecture \cite{federvardi}. This question was open for almost twenty years and indicates the difficulty of proving such dichotomy results, at least for decision problems. In counting complexity, however, it seems that obtaining
such results is less cumbersome. One reason for this is the existence of some powerful techniques like polynomial interpolation \cite{interpolation}, the Holant framework \cite{accidentalalgos,holographicalgos,fibonaccigates} as well as the principle of inclusion-exclusion which all have been used to establish very revealing dichotomy results such as \cite{holantdicho,embsdicho}. \\
Examples of dichotomies in parameterized counting complexity are the complete classifications of the homomorphism counting problem due to Dalmau and Jonsson \cite{homsdicho1}\footnote{Ultimately, the results of \cite{hombasis2017} and this work rely on the dichotomy for counting homomorphisms} and the subgraph counting problem due to Curticapean and Marx \cite{embsdicho}. For the latter, one is given graphs $H$ and $G$ and wants to count the number of subgraphs of $G$ isomorphic to $H$, parameterized by the size of $H$. It is known that this problem is polynomial time solvable if there is a constant upper bound on the size of the largest matching of $H$ and $\#\W$-hard otherwise\footnote{On the other hand the complexity of the decision version of this problem, that is, finding a subgraph of $G$ isomorphic to $H$, is still unresolved. Only recently it was shown in a major breakthrough that finding bicliques is hard \cite{biclique}.}. The first step in this proof was the hardness result of counting matchings of size $k$ of Curticapean \cite{kmatchings}, which turned out to be the ``bottleneck'' problem and was then reduced to the general problem.\\
This approach, first finding the hard obstructions and then reducing to the general case, seemed to be the canonical way to tackle such problems. However, recently Curticapean, Dell and Marx \cite{hombasis2017} discovered that a result of Lov{\'a}sz \cite{lovasz} implies the existence of parameterized reductions that, inter alia, allow a far easier proof of the general subgraph counting problem. Lov{\'a}sz result states that, given simple graphs $H$ and $G$, it holds that
\begin{equation}
\label{eqn:intro_lovasz}
 \#\embs(H,G) = \sum_{\rho \geq \emptyset} \mu(\emptyset,\rho) \cdot  \#\homs(H/\rho,G) \,,
\end{equation}
where the sum is over the elements of the partition lattice of $V(H)$, $\embs(H,G)$ is the set of embeddings\footnote{Note that embeddings and subgraphs are equal up to automorphisms, that is, counting embeddings and counting subgraphs are essentially the same problem.} from $H$ to $G$ and $\homs(H/\rho,G)$ is the set of homomorphisms from the graph $H/\rho$ obtained from $H$ by identifying vertices along $\rho$ to $G$. Furthermore $\mu$ is the Möbius function. In their work Curticapean, Dell and Marx showed in a general theorem that a summation $\sum_{i=1}^\ell c_i \cdot \#\homs(H_i,G)$ for \emph{pairwise non-isomorphic graphs $H_i$} is $\#\W$-hard if there is no upper bound on the treewidth of the pattern graphs $H_i$ and fixed-parameter tractable otherwise, using a dichotomy for counting homomorphisms due to Dalmau and Jonsson \cite{homsdicho1}. Having this, one only has to show two properties of (\ref{eqn:intro_lovasz}) to obtain the dichotomy for $\#\embs$. First, one has to show that a high matching number of $H$ implies that one of the graphs $H/\rho$ has high treewidth and second, that two (or more) terms with high treewidth and isomorphic graphs $H/\rho$ and $H/\sigma$ do not cancel out (note that the Möbius function can be negative). As there is a closed form for the Möbius function over the partition lattice it was possible to show that whenever $H/\rho$ and $H/\sigma$ are isomorphic the sign of the Möbius function is equal. 
\subsection{Our results}
The motivation of this work is the question whether the result of Curticapean, Dell and Marx can be generalized to construct a framework for the complexity classification of counting problems that can be expressed as the summation over homomorphisms and it turns out that this is possible whenever the summation is over a the lattice of a graphic matroid and the coefficients are evaluations of the Möbius function over the lattice, capturing not only embeddings but also locally injective homomorphisms.\\
In Section~\ref{sec:graphic_homs} we introduce what we call \emph{graphically restricted homomorphisms}: Intuitively, a graphical restriction $\tau(H)$ of a graph $H$ is a set of forbidden binary vertex identifications of $H$, modeled as a graph with vertex set $V(H)$ and edges along the binary constraints. We write $\taumins(H)$ as the set of all graphs obtained from $H$ by contracting vertices along edges in $\tau(H)$ and deleting multiedges, excluding those that contain selfloops. Now a graphically restricted homomorphism from $H$ to $G$ with respect to $\tau$ is a homomorphism from $H$ to $G$ that maps every pair of vertices $u,v \in V(H)$ that are adjacent in $\tau(H)$ to different vertices in $G$. We write $\tauhoms(H,G)$ for the set of all graphically restricted homomorphisms w.r.t. $\tau$ from $H$ to $G$ and provide a complete complexity classification for counting graphically restricted homomorphisms:

\begin{theorem}[Intuitive version]
	\label{thm:int_graphic_dicho}
Computing $\#\tauhoms(H,G)$ is fixed-parameter tractable when parameterized by $|V(H)|$ if the treewidth of every graph in $\taumins(H)$ is small. Otherwise the problem is $\#\W$-hard. 
\end{theorem}

In particular, we obtain the following algorithmic result:

\begin{theorem}
	\label{thm:intro_graphic_algo}
There exists a deterministic algorithm that computes $\#\tauhoms(H,G)$ in time $g(|V(H)|) \cdot |V(G)|^{\mathsf{tw}(\taumins(H)) + 1}$, where $g$ is a computable function and $\mathsf{tw}(\taumins(H))$ is the maximum treewidth of every graph in $\taumins(H)$.
\end{theorem}

Having established the general dichotomy we observe that there exist graphical restrictions $\tau_\mathsf{clique}$ and $\tauli$ such that $\homs_{\tau_\mathsf{clique}}(H,G)$ is the set of all subgraph embeddings from $H$ to $G$ and $\taulihoms(H,G)$ is the set of all \emph{locally injective homomorphisms} from $H$ to $G$.\\
As a consequence we obtain a full complexity dichotomy for counting locally injective homomorphisms from small pattern graphs $H$ to big host graphs $G$. To the best of the author's knowledge, this is the first result about the complexity of counting locally injective homomorphisms.

\begin{corollary}[Intuitive version]
	\label{cor:int_li_dicho}
Computing the number of locally injective homomorphisms from $H$ to $G$ is fixed-parameter tractable when parameterized by $|V(H)|$ if the treewidth of every graph in $\taulimins(H)$ is small. Otherwise the problem is $\#\W$-hard.\\ Moreover, there exists a deterministic algorithm that computes this number in time $g(|V(H)|) \cdot |V(G)|^{\mathsf{tw}(\taulimins(H)) + 1}$, where $g$ is a computable function and $\mathsf{tw}(\taulimins(H))$ is the maximum treewidth of every graph in $\taulimins(H)$.
\end{corollary}

We then observe that --- in contrast to subgraph embeddings --- counting locally injective homomorphisms has ``real'' FPT cases, that is, cases that are fixed-parameter tractable but not polynomial time solvable under standard assumptions. We show this by restricting the pattern graph to be a tree:

\begin{corollary}
	\label{cor:intro_dicho_trees}
Computing the number of locally injective homomorphisms from a tree $T$ to a graph $G$ can be done in deterministic time $g(|V(T)|) \cdot |V(G)|^2$, that is, the problem is fixed-parameter tractable when parameterized by $|V(T)|$. On the other hand, the problem is $\#\P$-hard.
\end{corollary}

To prove $\#\P$-hardness, we prove in an intermediate step that the subgraph counting problem remains hard when both graphs are restricted to be trees, which may be of independent interest:

\begin{lemma}
	\label{lem:into_classical_trees}
The problem of, given trees $T_1$ and $T_2$, computing the number of subtrees of $T_2$ that are isomorphic to $T_1$ is $\#\P$-hard.
\end{lemma}

After that we generalize locally injective homomorphisms to homomorphisms that are injective in the $r$-neighborhood of every vertex and observe that those are also graphically restricted and consequently obtain a counting dichotomy as well. \\

Finally, we show in Section~\ref{sec:lincombs} that all results can easily be extended to so-called \emph{linear combinations} of graphically restricted homomorphisms. Here one gets as input graphs $H_1,\dots,H_\ell$ together with positive coefficients $c_1,\dots,c_\ell$ and a graph $G$ and the goal is to compute
\[ \sum_{i=1}^\ell c_i \cdot \#\homs_{\tau_i}(H_i,G) \,,\]
for graphical restrictions $\tau_1,\dots,\tau_\ell$. This generalizes for example problems like counting all trees of size $k$ in $G$ or counting all locally injective homomorphisms from all graphs of size $k$ to $G$ or a combination thereof. We find out that, under some conditions, the dichotomy criteria transfer immediately to linear combinations:
\begin{theorem}[Intuitive version]
\label{thm:int_graphic_lincombs}
Computing $\sum_{i=1}^\ell c_i \cdot \#\homs_{\tau_i}(H_i,G)$ is fixed-parameter tractable when parameterized by $\max_i\{|V(H_i)|\}$ if the maximum treewidth of every graph in $\bigcup_i^\ell \tau_i\text{-}\mathcal{M}(H_i)$ is small. Otherwise, if additionally $|V(H_i)|$ has the same parity for every $i \in [\ell]$, the problem is $\#\W$-hard.
\end{theorem}

Furthermore we observe that this theorem is not true on the $\#\W$-hardness side if we omit the parity condition.

\subsection{Techniques}
The main ingredients of the proofs of Theorem~\ref{thm:int_graphic_dicho} and Theorem~\ref{thm:intro_graphic_algo} are the complexity classification of linear combinations of homomorphisms due to Curticapean, Dell and Marx (see Lemma~3.5 and Lemma~3.8 in \cite{hombasis2017}) as well as a corollary of Rota's NBC Theorem (see e.g. Theorem~4 in \cite{rota1964foundations}). In the first step we prove the following identity for the number of graphically restricted homomorphisms via Möbius inversion:

\[ \#\tauhoms(H,G) = \sum_{\rho \geq \emptyset} \mu(\emptyset,\rho)\cdot \#\homs(H/\rho,G) \,,\]
where the sum is over elements of the lattice of flats of the graphical matroid given by $\tau(H)$ and $H/\rho$ is the graph obtained by contracting the vertices of $H$ along the flat $\rho$. After that we use Rota's Theorem to prove that none of the terms cancel out\footnote{Here ``cancel out'' means that it could be possible that $H/\rho$ and $H/\sigma$ are isomorphic, but $\mu(\emptyset,\rho) = - \mu(\emptyset,\sigma)$ and all other $H/\rho'$ are not isomorphic to $H/\rho$. In this case, the term $\#\homs(H/\rho,G)$ would vanish in the above identity.}, despite the fact that the Möbius function can be negative. More precisely we show that whenever $H/\rho \cong H/\sigma$, we have that $\mathsf{rk}(\rho) = \mathsf{rk}(\sigma)$ and therefore, by Rota's Theorem, $\mathsf{sgn}(\mu(\emptyset,\rho)) = \mathsf{sgn}(\mu(\emptyset,\sigma))$.\\
The dichotomies for locally injective homomorphisms and homomorphisms that are injective in the $r$-neighborhood of every vertex are mere applications of the general theorem. For $\#\P$-hardness of the subgraph counting problem restricted to trees, we adapt the idea of the ``skeleton graph'' by Goldberg and Jerrum \cite{goldberg_jerrum_trees} and reduce directly from computing the permanent. To transfer this result to locally injective homomorphisms we use the well-known observation that locally injective homomorphisms from a tree to a tree are embeddings.\\
Finally, we prove the dichotomy for linear combinations of graphically restricted homomorphisms by taking a closer look at the proof of Theorem~\ref{thm:int_graphic_dicho}. Here, the parity constraint of the vertices of the graphs in the linear combination assures that there are no graphs $H_i$ and $H_j$ and elements $\rho_i$ and $\rho_j$ of the matroid lattices of $\tau_i(H_i)$ and $\tau_j(H_j)$ such that $H_i/\rho_i$ and $H_j/\rho_j$ are isomorphic but $\rho_i$ and $\rho_j$ have ranks of different parities. Using this observation, Theorem~\ref{thm:int_graphic_lincombs} can be proven in the same spirit as Theorem~\ref{thm:int_graphic_dicho}.

\section{Preliminaries}
\label{sec:prelims}
First we will introduce some basic notions: Given a finite set $S$, we write $|S|$ or $\#S$ for the cardinality of $S$. Given a natural number $\ell$ we let $[\ell]$ be the set $\bs{1,\dots,\ell}$. Given a real number $r$ we define the \emph{sign} $\mathsf{sgn}(r)$ of $r$ to be $1$ if $r > 0$, $0$ if $r = 0$ and $-1$ if $r < 0$.\\
A \emph{poset} is a pair $(P,\leq)$ where $P$ is a set and $\leq$ is a binary relation on $P$ that is reflexive, transitive and anti-symmetric. Throughout this paper we will write $y \geq x$ if $x \leq y$. A~\emph{lattice} is a poset $(L,\leq)$ such that every pair of elements $x,y \in L$ has a \emph{least upper bound} $x \vee y$ and a \emph{greatest lower bound} $x \wedge y$ that satisfy:
\begin{itemize}
\item $x \vee y \geq x$, $x \vee y \geq y$ and for all $z$ such that $z \geq x$ and $z \geq y$ it holds that $z \geq x \vee y$.
\item $x \wedge y \leq x$, $x \wedge y \leq y$ and for all $z$ such that $z \leq x$ and $z \leq y$ it holds that $z \leq x \wedge y$.
\end{itemize}
Given a finite set $S$, a \emph{partition} of $S$ is a set $\rho$ of pairwise disjoint subsets of $S$ such that $\dot{\bigcup}_{s \in \rho} s = S$. We call the elements of $\rho$ \emph{blocks}. For two partitions $\rho$ and $\sigma$ we write $\rho \leq \sigma$ if every element of $\rho$ is a subset of some element of $\sigma$. This binary relation is a lattice and called the \emph{partition lattice} of $S$. We will in particular encounter lattices of graphic matroids in our proofs. 
\subsection{Matroids}
We will follow the definitions of Chapt.~1 of the textbook of Oxley \cite{oxley}.
\begin{definition}
A \emph{matroid} $M$ is a pair $(E,\mathcal{I})$ where $E$ is a finite set and $\mathcal{I}\subseteq \mathcal{P}(E)$ such that
\begin{enumerate}
\item[(1)] $\emptyset \in \mathcal{I}$,
\item[(2)] if $A \in \mathcal{I}$ and $B \subseteq A$ then $B \in \mathcal{I}$, and
\item[(3)] if $A,B \in \mathcal{I}$ and $|B|<|A|$ then there exists $a \in A \setminus B$ such that $B\cup \bs{a} \in \mathcal{I}$.
\end{enumerate} 
We call $E$ the \emph{ground set} and an element $A \in \mathcal{I}$ an \emph{independent set}. A maximal independent set is called a \emph{basis}. The \emph{rank} $\mathsf{rk}(M)$ of $M$ is the size of its bases\footnote{This is well-defined as every maximal independent set has the same size due to (3).}.
\end{definition}

Given a subset $X \subseteq E$ we define $\mathcal{I}|X := \bs{A \subseteq X ~|~A \in \mathcal{I}}$. Then $M|X := (X,\mathcal{I}|X)$ is also a matroid and called the restriction of $M$ to $X$. Now the \emph{rank} $\mathsf{rk}(X)$ of $X$ is the rank of $M|X$. Equivalently, the rank of $X$ is the size of the largest independent set $A \subseteq X$.\\
Furthermore we define the \emph{closure} of $X$ as follows:
\[ \mathsf{cl}(X) := \bs{e \in E ~|~ \mathsf{rk}(X \cup \bs{e}) = \mathsf{rk}(X)} \,.\]
Note that by definition $\mathsf{rk}(X)=\mathsf{rk}(\mathsf{cl}(X))$. We say that $X$ is a \emph{flat} if $\mathsf{cl}(X)=X$. We denote $L(M)$ as the set of flats of $M$. It holds that $L(M)$ together with the relation of inclusion is a lattice, called the \emph{lattice of flats} of $M$. The least upper bound of two flats $X$ and $Y$ is $\mathsf{cl}(X \cup Y)$ and the greatest lower bound is $X \cap Y$. It is known that the lattices of flats of matroids are exactly the geometric lattices\footnote{For the purpose of this paper we do not need the definition of geometric lattices but rather the equivalent one in terms of lattices of flats and therefore omit it. We recommend e.g. Chapt.~3 of \cite{welsh} and Chapt.~1.7 of \cite{oxley} to the interested reader.} and we denote the set of those lattices as $\mathcal{L}$.

In Section~\ref{sec:graphic_homs} we take a closer look at (lattices of flats of) graphic matroids:
\begin{definition}
Given a graph $H=(V,E) \in \graphs$, the \emph{graphic matroid} $M(H)$ has ground set $E$ and a set of edges is independent if and only if it does not contain a cycle.
\end{definition}
If $H$ is connected then a basis of $H$ is a spanning tree of $H$. If $H$ consists of several connected components then a basis of $M(H)$ induces spanning trees for each of those. Every subset $X$ of $E$ induces a partition of the vertices of $H$ where the blocks are the vertices of the connected components of $H|_X$ and it holds that 
\begin{equation}
\label{eqn:graphic_rank}
\mathsf{rk}(X) = |V(H)| - c(H|_X) \,.
\end{equation}
In particular, the flats of $M(H)$ correspond bijectively to the partitions of vertices of $H$ into connected components as adding an element to $X$ such that the rank does not change will not change the connected components, too. For convenience we will therefore abuse notation and say, given an element $\rho$ of the lattice of flats of $M(H)$, that $\rho$ partitions the vertices of $H$ where the blocks are the vertices of the connected components of $H|_\rho$. The following observation will be useful in Section~\ref{sec:graphic_homs}:

\begin{lemma}
\label{lem:same_blocks_same_ranks}
Let $\rho,\sigma \in L(M(H))$ for a graph $H\in \graphs$. If the number of blocks of $\rho$ and $\sigma$ are equal then $\mathsf{rk}(\rho)=\mathsf{rk}(\sigma)$. 
\end{lemma}

\begin{proof}
Immediately follows from Equation~(\ref{eqn:graphic_rank}).
\end{proof}

We denote $H/\rho$ as the graph obtained from $H$ by contracting the vertices of $H$ that are in the same component of $\rho$ and deleting multiedges (but keeping selfloops). As the vertices of $H/\rho$ partition the vertices of $H$, we think of the vertices of $H/\rho$ as subsets of vertices of $H$ and call them \emph{blocks}. Furthermore we write $[v]$ for the block containing~$v$.

\subsection{Graphs and homomorphisms}

In this work all graphs are considered unlabeled and simple but may allow selfloops unless stated otherwise. We denote the set of all those graphs as $\loopG$. Furthermore we denote $\graphs$ as the set of all unlabeled and simple graphs without selfloops. \\
For a graph $G$ we write $n$ for the number of vertices $V(G)$ of $G$ and $m$ for the number of edges $E(G)$ of $G$. We denote $c(G)$ as the number of connected components of $G$. Furthermore, given a subset $X$ of edges, we denote $G|_X$ as the graph with vertices $V(G)$ and edges $X$. Given a partition of vertices $\rho$ of a graph $H$, we write $H/\rho$ as the graph obtained from $H$ by contracting the vertices of $H$ that are in the same component of $\rho$ and deleting multiedges (but keeping selfloops). As the vertices of $H/\rho$ partition the vertices of $H$, we think of the vertices as subsets of vertices of $H$ and call them \emph{blocks}. Furthermore we write $[v]$ for the block containing~$v$. \\
Given graphs $H$ and $G$, a \emph{homomorphism} from $H$ to $G$ is a mapping $\varphi : V(H) \rightarrow V(G)$ such that $\bs{u,v} \in E(H)$ implies that $\bs{\varphi(u),\varphi(v)} \in E(G)$. We denote $\homs(H,G)$ as the set of all homomorphisms from $H$ to $G$. A homomorphism is called \emph{embedding} if it is injective and we denote $\embs(H,G)$ as the set of all embeddings from $H$ to $G$. An embedding from $H$ to $H$ is called an \emph{automorphism} of $H$. We denote $\mathsf{Aut}(H)$ as the set of all automorphisms of $H$. Furthermore we let $\mathsf{Sub}(H,G)$ be the set of all subgraphs of $G$ that are isomorphic to $H$. Then it holds that $\#\mathsf{Aut}(H)\cdot\#\mathsf{Sub}(H,G) = \#\embs(H,G)$ (see e.g. \cite{lovasz}). \\
Given a set $S$ and a function $\alpha : S \rightarrow \Q$, we define the \emph{support} of $\alpha$ as follows:
\[ \mathsf{supp}(\alpha) := \bs{s \in S ~|~ \alpha(s)\neq 0 } \,.\]

A graph parameter that will be of quite some importance to define the dichotomy criteria is the \emph{treewidth} of a graph, capturing how ``tree-like'' a graph is:
\begin{definition}[Chapt.~7 in \cite{cygan2015parameterized}]
	A \emph{tree decomposition} of a graph $G \in \graphs$ is a pair $\mathcal{T}=(T,\bs{X_t}_{t_\in V(T)})$, where $T$ is a tree whose every node $t$ is assigned a vertex subset $X_t \subseteq V(G)$, such that:
	\begin{enumerate}
		\item[(1)] $\bigcup_{t \in V(T)} X_t = V(G)$.
		\item[(2)] For every $\bs{u,v} \in E(G)$, there exists $t \in V(T)$ such that $u$ and $v$ are contained in $X_t$.
		\item[(3)] For every $u \in V(G)$, the set $T_u:=\bs{t \in V(T)~|~u \in X_t}$ induces a connected subtree of $T$.
	\end{enumerate}
The \emph{width} of $\mathcal{T}$ is the size of the largest $X_t$ for $t \in V(T)$ minus $1$ and the \emph{treewidth} of $G$ is the minimum width of any tree decomposition of $G$. We write $\mathsf{tw}(G)$ for the treewidth of $G$. Given a finite set of graphs $\mathcal{M}$, we denote $\mathsf{tw}(\mathcal{M})$ as the maximum treewidth of any graph in $\mathcal{M}$.
\end{definition}

Examples of graphs with small treewidth are matchings, paths and more generally trees and forests or cycles. On the other hand, graphs with high treewidth are for example cliques, bicliques and grid graphs.\\
Throughout this paper we will often say that a set $C$ of graphs has \emph{bounded treewidth} meaning that there is a constant $B$ such that the treewidth of every graph $H\in C$ is bounded by $B$.

\subsection{Parameterized counting}
We will mainly follow the definitions of Chapt.~14 of the textbook of Flum and Grohe \cite{flumgrohe_counting}. A \emph{parameterized counting problem} is a function $F:\bs{0,1}^\ast \rightarrow \N$ together with a polynomial-time computable \emph{parameterization} $k:\bs{0,1}^\ast \rightarrow \N$. A parameterized counting problem is \emph{fixed-parameter tractable} if there exists a computable function $g$ such that it can be solved in time $g(k(x))\cdot |x|^{O(1)}$ for any input $x$. A \emph{parameterized Turing reduction} from $(F,k)$ to $(F',k')$ is an FPT algorithm w.r.t. parameterization $k$ with oracle $(F',k')$ that on input $x$ computes $F(x)$ and additionally satisfies that there exists a function $g'$ such that for every oracle query $y$ it holds that $k'(y)\leq g(k(x))$. A parameterized counting problem $(F,k)$ is $\#\W$-hard if there exists an FPT Turing reduction from $\#k$-$\mathsf{clique}$ to $(F,k)$, where $\#k$-$\mathsf{clique}$ is the problem of, given a graph $G$ and a parameter $k$, computing the number of cliques of size $k$ in $G$\footnote{For a more detailed introduction to $\#\W$ we recommend \cite{flumgrohe_counting} to the interested reader.}. Under standard assumptions (e.g. under the exponential time hypothesis) $\#\W$-hard problems are not fixed-parameter tractable.\\ 
The following two parameterized counting problems will be of particular importance in this work: Given a class of graphs $C\subseteq \graphs$, $\#\homs(C)$ ($\#\embs(C)$) is the problem of, given a graph $H \in C$ and a graph $G \in \graphs$, computing $\#\homs(H,G)$ ($\#\embs(H,G)$). Both problems are parameterized by $\#V(H)$. Their complexity has already been classified:

\begin{theorem}[\cite{homsdicho1}]
\label{thm:hom-dichotomy}
Let $C$ be a recursively enumerable class of graphs. If $C$ has {bounded treewdith} then $\#\homs(C)$ can be solved in polynomial time. Otherwise $\#\homs(C)$ is $\#\W$-hard.
\end{theorem}

\begin{theorem}[\cite{embsdicho}]
\label{thm:emb-dichotomy}
Let $C$ be a recursively enumerable class of graphs. If $C$ has {bounded matching number} then $\#\embs(C)$ can be solved in polynomial time. Otherwise $\#\embs(C)$ is $\#\W$-hard.
\end{theorem}

Recall that ``bounded treewidth (matching number)'' means that there is a constant $B$ such that the treewidth (size of the largest matching) of any graph in $C$ is bounded by $B$.

\subsection{Linear combinations of homomorphisms and Möbius inversion}
Curticapean, Dell and Marx \cite{hombasis2017} introduced the following parameterized counting problem:
\begin{definition}[Linear combinations of homomorphisms]
Let $\mathcal{A}$ be a set of functions $a:~\mathcal{G}~\rightarrow~\mathbb{Q}$ with finite support\footnote{We can also think of $\mathcal{A}$ being a set of lists.}. We define the parameterized counting problem $\#\homs(\mathcal{A})$ as follows:\\
Given $a \in \mathcal{A}$ and $G \in \mathcal{G}$, compute
\[ \sum_{H \in \mathrm{supp}(a)} a(H) \cdot \#\homs(H, G) \,, ~\text{  parameterized by } \max_{H \in \mathrm{supp}(a)} \#V(H) \,.\] 
\end{definition}
Note that this problem generalizes $\#\homs(C)$.
The following theorem will be the foundation of all complexity results in this paper:
\begin{theorem}[\cite{hombasis2017}, Lemma~3.5 and Lemma~3.8]
\label{thm:lincombs_homs}
If $\mathcal{A}$ has bounded treewidth then $\#\homs(\mathcal{A})$ can be solved in time $g(|\mathrm{supp}(\alpha)|) \cdot n^{O(1)}$ on input $(\alpha,G)$ where $n = |V(G)|$ and $g$ is a computable function. Otherwise the problem is $\#\W$-hard.
\end{theorem}

In their paper, the authors show how this result can be used to give a much simpler proof of Theorem~\ref{thm:emb-dichotomy}. The idea is that every problem $\#\embs(C)$ is equivalent to a problem $\#\homs(\mathcal{A})$. As all proofs in this work are in the same flavour, we will outline the technique here, using $\#\embs(C)$ as an example. Therefore, we first need to introduce the so called Möbius inversion (we recommend reading \cite{stanley2011enumerative} for a more detailed introduction):\\

\begin{definition}
\label{def:zeta_trans}
Let $(P,\leq)$ be a poset and $h:P \rightarrow \mathbb{C}$ be a function. Then the \emph{zeta transformation} $\zeta h$ is defined as follows:
\[ \zeta h (\sigma) := \sum_{\rho \geq \sigma} h(\rho) \,.\]
\end{definition}

\begin{theorem}[Möbius inversion, see \cite{stanley2011enumerative} or \cite{rota1964foundations}]
\label{thm:möbi}
Let $(P,\leq)$ and $h$ as in Definition~\ref{def:zeta_trans}. Then there is a function $\mu_P: P \times P \rightarrow \mathbb{Z}$ such that for all $\sigma \in P$ it holds that
\[ h(\sigma) = \sum_{\rho \geq \sigma} \mu_P(\sigma,\rho) \cdot \zeta h (\rho) \,.\]
$\mu_P$ is called the \emph{Möbius function}.
\end{theorem}
The following identity is due to Lov{\'a}sz \cite{lovasz}:
\[ \#\homs(H/\sigma,G) = \sum_{\rho \geq \sigma} \#\embs(H/\rho,G) \,,\]
where $\sigma$ and $\rho$ are partitions of vertices of $H$ and $\geq$ is the partition lattice of $H$. Now Möbius inversion yields the following identity \cite{lovasz}:
\[ \#\embs(H,G) = \sum_{\rho \geq \emptyset} \mu(\emptyset,\rho) \cdot  \#\homs(H/\rho,G) \,,\]
where $\mu$ is the Möbius function over the partition lattice. Therefore, for every class of graphs $C$, there is a family of functions with finite support $\mathcal{A}$ such that $\#\embs(C)$ and $\#\homs(\mathcal{A})$ are the same problems. Now Curticapean, Dell and Marx show that $C$ has unbounded matching number if and only if $\mathcal{A}$ has unbounded treewidth. The critical point in this proof was to show that the sign of $\mu(\emptyset,\rho)$ only depends on the number of blocks of $\rho$, which implies that for two isomorphic graphs $H_1$ and $H_2$, the terms $\#\homs(H_1,G)$ and $\#\homs(H_2,G)$ have the same sign in the above identity and therefore do not cancel out in the homomorphism basis. As there is a closed form for $\mu(\emptyset,\rho)$\footnote{Here it is crucial that $\mu$ is the Möbius function over the (complete) partition lattice.}, the information about the sign could easily be extracted. \\
The motivation of this work is the question whether this can be made more general and it turns out that a corollary of Rota's NBC Theorem \cite{rota1964foundations} (see also \cite{blass1997mobius}) captures exactly what we need:
\begin{theorem}[See e.g. Theorem~4 in \cite{rota1964foundations}]
\label{thm:rota}
Let $L$ be a geometric lattice with unique minimal element $\bot$ and let $\rho$ be an element of $L$. Then it holds that \[\mathsf{sgn}(\mu_L(\bot,\rho)) = (-1)^{\mathsf{rk}(\rho)} \,.\]
\end{theorem} 

In the following we will show that combining Rota's Theorem and the dichotomy for counting linear combinations of homomorphisms yields complete complexity classifications for the problems of counting those restricted homomorphisms that induce a Möbius inversion over the lattice of a graphic matroid, which are known to be geometric, when transformed into the homomorphism basis. Those include embeddings as well as locally injective homomorphisms.

\section{Graphically restricted homomorphisms}
\label{sec:graphic_homs}
In the following we write $\emptyset$ for the minimal element of a matroid lattice.

\begin{definition}
	\label{dfn:graphic_homs}
A \emph{graphical restriction} is a computable mapping $\tau$ that maps a graph $H \in \graphs$ to a graph $H' \in \graphs$ such that
$V(H)=V(H')$, that is, $\tau$ only modifies edges of $H$. We denote the set of all graphical restrictions as $\Tau$. Given graphs $H$ and $G$ and a graphical restriction $\tau$, we define the set of \emph{graphically restricted homomorphisms} w.r.t. $\tau$ from $H$ to $G$ as follows:
\[ \tauhoms(H,G) := \bs{\varphi \in \homs(H,G) ~ |~\forall u,v \in V(H): \bs{u,v} \in E(\tau(H)) \Rightarrow \varphi(u) \neq \varphi(v) } \,.\]
Given a recursively enumerable class of graphs $C \subseteq \graphs$, we define the parameterized counting problem $\#\tauhoms(C)$ as follows: Given a graph $H \in C$ and a graph $G \in \graphs$, we parameterize by $|V(H)|$ and wish to compute $\#\tauhoms(H,G)$.
\end{definition}

Assume for example that $\tau_\mathsf{clique}$ maps a graph $H$ to the complete graph with vertices $V(H)$. Then one can easily verify that $\homs_{\tau_\mathsf{clique}}(H,G) = \embs(H,G)$.

The following lemma is an application of Möbius inversion (and slightly generalizes \cite{lovasz}).

\begin{lemma}
\label{lem:main_graphic}
Let $\tau$ be a graphical restriction. Then for all graphs $H \in \loopG$ and $G \in \graphs$ it holds that
\begin{equation}
\#\tauhoms(H,G)=\sum_{\rho \geq \emptyset} \mu(\emptyset,\rho)\cdot \#\homs(H/\rho,G) \,,
\end{equation}
where $\leq$ and $\mu$ are the relation and the Möbius function of the lattice $L(M(\tau(H)))$.
\end{lemma}
\begin{proof}
Let $\tau$ and $H$ be fixed and let $\homs(H/\rho,G)[\tau]$ be the set of all homomorphisms $\varphi \in \homs(H/\rho,G)$ such that $\bs{u,v} \in E(\tau(H))$ and $[u]\neq[v]$ imply that $\varphi([u]) \neq \varphi([v])$. More precisely: 
\begin{align*}
~&\homs(H/\rho,G)[\tau] := \\
~& \bs{\varphi \in \homs(H/\rho,G) ~|~ \forall u,v \in V(H): \bs{u,v} \in E(\tau(H)) \wedge [u]\neq[v] \Rightarrow \varphi([u]) \neq \varphi([v]) } \,.
\end{align*}
We will first prove the following identities:
\begin{claim}
\label{clm:minimal_id}
For all $G \in \graphs$ it holds that
\[ \#\homs(H/\emptyset,G)[\tau] = \#\tauhoms(H,G) \,. \]
\end{claim}
\begin{proof}
Every block in $H/\emptyset$ is a singleton and $H \cong H/\emptyset$. Now the identity trivially follows from Definition~\ref{dfn:graphic_homs}. 
\end{proof}

\begin{claim}
\label{clm:zeta_matroid}
For all $G \in \graphs$ and $\sigma \in L(M(\tau(H)))$ it holds that
\begin{equation}
\#\homs(H/\sigma,G) = \sum_{\rho \geq \sigma} \#\homs(H/\rho,G)[\tau]
\end{equation}
\end{claim}
\begin{proof}
Let $[v]$ be the block of $v$ in $H/\sigma$.
We define an equivalence relation $\sim_\tau$ over $\homs(H/\sigma,G)$ as follows:
\[ \varphi \sim_\tau \psi :\Leftrightarrow \forall \bs{u,v}\in E(\tau(H)): \varphi([u])=\varphi([v]) \Leftrightarrow \psi([u])=\psi([v]) \,. \]
We write $[\varphi]_\tau$ for the equivalence class of $\varphi$ and let $H/{[\varphi]_\tau}$ be the graph obtained from $H/\sigma$ by further contracting different blocks $[u]$ and $[v]$ whenever $\bs{u,v}\in E(\tau(H))$ and $\varphi([u])=\varphi([v])$ (note that this is well-defined by the definition of $\sim_\tau$). Now consider $\sigma$ in the graphical matroid $M(\tau(H))$. Every block $[v]$ corresponds to a connected component of the flat given by $\sigma$. Now contracting different blocks $[u]$ and $[v]$ for $\bs{u,v}$ in $E(\tau(H))$ is a refinement of $\sigma$ obtained by adding the edge $\bs{u,v}$ in $M(\tau(H))$ and taking the closure. Therefore the equivalence classes of $\sim_\tau$ and the refinements of $\sigma$ in the matroid lattice correspond bijectively and we write $[\rho]_\tau$ for the equivalence class corresponding to $\rho$. It remains to show that for every $\rho \geq \sigma$ we have that
\begin{equation}
\label{eqn:number_equi}
 |[\rho]_\tau| = \#\homs(H/\rho,G)[\tau] \,.
\end{equation}
This can be proven by constructing a bijection $b$. We write $[v]_\sigma$ for blocks in $H/\sigma$ and $[v]_\rho$ for blocks in $H/\rho$. On input $\varphi \in [\rho]_\tau$, $b$ outputs the homomorphism in $\homs(H/\rho,G)[\tau]$ that maps a block $[v]_\rho$ to $\varphi([v]_\sigma)$. This is well-defined as $\varphi$ maps blocks $[u]_\sigma$ and $[v]_\sigma$ to the same vertex in $G$ if and only if they are subsumed by a common block in $H/\rho$ (recall that $\rho \geq \sigma$ in the matroid lattice). On the other hand we can construct a mapping $b'$ that given $\psi \in \homs(H/\rho,G)[\tau]$ outputs the homomorphism in $[\rho]_\tau$ that maps a block $[v]_\sigma$ to the image of the block $[v]_\rho$ (that subsumes $[v]_\sigma$) according to $\psi$. Now $b \circ b' = \mathsf{id}_{\homs(H/\rho,G)[\tau]}$ and $b' \circ b = \mathsf{id}_{[\rho]_\tau}$. Consequently, $b$ is a bijection and Equation~\ref{eqn:number_equi} holds.
Now we have
\begin{align*}
~&\#\homs(H/\sigma,G) \\~&= \left|\dot{\bigcup}_{[\varphi]_\tau \in \homs(H/\sigma,G)/\sim_\tau } [\varphi]_\tau \right| = \left|\dot{\bigcup}_{\rho \geq \sigma } [\rho]_\tau \right|= \sum_{\rho \geq \sigma} \left|[\rho]_\tau \right| =\sum_{\rho \geq \sigma} \#\homs(H/\rho,G)[\tau]
\end{align*}
which proves the claim.
\end{proof}
Now Claim~\ref{clm:zeta_matroid} is a zeta transform over the matroid lattice of $M(\tau(H))$. By Möbius inversion (Theorem~\ref{thm:möbi}) we obtain that
\[ \#\homs(H/\emptyset,G)[\tau] = \sum_{\rho \geq \emptyset} \mu(\emptyset,\rho) \cdot \#\homs(H/\rho,G) \,,\]
and hence, by Claim~\ref{clm:minimal_id}, 
\[ \#\tauhoms(H,G) = \sum_{\rho \geq \emptyset} \mu(\emptyset,\rho) \cdot \#\homs(H/\rho,G) \,. \]
\end{proof}

Intuitively, we will now show that counting graphically restricted homomorphisms from $H$ to $G$ is hard if we can ''glue'' vertices of $H$ together along edges of $\tau(H)$ such that the resulting graph has no selfloops and high treewidth. We will capture this intuition formally:

\begin{definition}
	Let $H \in \graphs$ be a graph and let $\tau$ be a graphical restriction. A graph $H'\in \loopG$ obtained from $H$ by contracting pairs of vertices $u$ and $v$ such that $\bs{u,v} \in E(\tau(H))$ and deleting multiedges (but keeping selfloops) is called a $\tau$-\emph{contraction} of $H$. If additionally $H' \in \graphs$, that is, the contraction did not yield selfloops, we call $H'$ a $\tau$-\emph{minor} of $H$. We denote the set of all $\tau$-minors of $H$ as $\taumins(H)$ and given a class of graphs $C \subseteq \graphs$ we denote the set of all $\tau$-minors of all graphs in $C$ as $\taumins(C)$.
\end{definition}

Finally, we can classify the complexity of counting graphically restricted homomorphisms along the treewidth of their $\tau$-minors:

\begin{theorem}[Theorem~\ref{thm:int_graphic_dicho} and Theorem~\ref{thm:intro_graphic_algo}, restated]
	\label{thm:graphic_dichotomy}
Let $\tau$ be a graphical restriction and let $C \subseteq \graphs$ be a recursively enumerable class of graphs. Then $\#\tauhoms(C)$ is FPT if $\taumins(C)$ has bounded treewidth and $\#\W$-hard otherwise. Furthermore, given $H,G \in \graphs$, there exists a deterministic algorithm that computes $\#\tauhoms(H,G)$ in time 
\[ g(|V(H)|) \cdot |V(G)|^{\mathsf{tw}(\taumins(H)) + 1} \,, \] where $g$ is a computable function.
\end{theorem}

\begin{proof}
By Lemma~\ref{lem:main_graphic} we have that 
\[ \#\tauhoms(H,G) = \sum_{\rho \geq \emptyset} \mu(\emptyset,\rho) \cdot \#\homs(H/\rho,G) \,. \]
Now, as $G$ has no selfloops, a term $\#\homs(H/\rho,G)$ is zero whenever $H/\rho$ has a selfloop. Consequently, for every non-zero term $\#\homs(H/\rho,G)$, it holds that $H/\rho \in \taumins(H)$. Therefore, by Lemma 3.5 in \cite{hombasis2017}, we obtain an algorithm computing $\#\tauhoms(H,G)$ in time
\[ g(|V(H)|) \cdot |V(G)|^{\mathsf{tw}(\taumins(H)) + 1} \,, \]
for a computable function $g$. This immediately implies that the problem $\#\tauhoms(C)$ is fixed-parameter tractable if $\taumins(C)$ has bounded treewidth. It remains to show that $\#\tauhoms(C)$ is $\#\W$-hard otherwise. By condensing all terms $\#\homs(H/\rho,G)$ and $\#\homs(H/\sigma,G)$ where $H/\rho$ and $H/\sigma$ are isomorphic, it follows that there exist coefficients $c_H[H']$ for every $H' \in \taumins(H)$ such that
\[ \#\tauhoms(H,G) = \sum_{H' \in \taumins(H)} c_H[H'] \cdot \#\homs(H',G) \,. \]
We will now show that none of the $c_H[H']$ is zero: It holds that
\begin{equation}
\label{eqn:coef_sum}
c_H[H'] = \sum_{\substack{\rho \geq \emptyset \\ H' \cong H/\rho}} \mu(\emptyset,\rho) \,.
\end{equation}
Consider $\rho$ and $\rho'$ such that $H/\rho \cong H/\rho' \cong H'$. It follows that \[\mathsf{rk}(\rho) = |V(H)| - c(H/\rho) =  |V(H)| - c(H') = |V(H)| - c(H/\rho') = \mathsf{rk}(\rho') \,. \] Now, as the lattice of $M(\tau(H))$ is geometric, we can apply the corollary of Rota's NBC Theorem (Theorem~\ref{thm:rota}) and obtain that $\mathsf{sgn}(\mu(\emptyset,\rho)) = (-1)^{\mathsf{rk}(\rho)} = (-1)^{\mathsf{rk}(\rho')} = \mathsf{sgn}(\mu(\emptyset,\rho'))$. Consequently every term in Equation~(\ref{eqn:coef_sum}) has the same sign and therefore $c_H[H']\neq 0$. Now we define a function $a_H:\graphs \rightarrow \mathbb{Q}$ as follows
\begin{equation*}
	a_H(F) := \begin{cases} c_H[F] &\text{if } F \in \taumins(H)\\
		0 &\text{otherwise}
		\end{cases}
\end{equation*}
and we set $\mathcal{A}_C = \bs{ a_H~|~H \in C}$. Then the problems $\#\homs(\mathcal{A}_C)$ and $\#\tauhoms(C)$ are equivalent w.r.t. parameterized turing reductions. As $c_H[H'] \neq 0$ for every $H' \in \taumins(H)$ it follows that $\mathcal{A}_C$ has unbounded treewidth if and only if $\taumins(C)$ has unbounded treewidth. We conclude by Theorem~\ref{thm:lincombs_homs} that $\#\tauhoms(C)$ is $\#\W$-hard in this case. 
\end{proof}

\section{Locally injective homomorphisms}

In this section we are going to apply the general dichotomy theorem to the concrete case of counting locally injective homomorphisms. A homomorphism $\varphi$ from $H$ to $G$ is \emph{locally injective} if for every $v \in V(H)$ it holds that $\varphi|_{N(v)}$ is injective. We denote $\lihoms(H,G)$ as the set of all locally injective homomorphisms from $H$ to $G$ and we define the corresponding counting problem $\#\lihoms(C)$ for a class of graphs $C \subseteq \graphs$ as follows: Given graphs $H \in C$ and $G \in \graphs$, compute $\#\lihoms(H,G)$. The parameter is $|V(H)|$. Locally injective homomorphisms have already been studied by Ne{\v{s}}et{\v{r}}il in 1971 \cite{nesetril} and were applied in the context of distance constrained labelings of graphs (see \cite{fiala} for an overview).
As well as subgraphs embeddings, locally injective homomorphisms are graphically restricted homomorphisms. 

\begin{lemma}
\label{lem:li_graphic}
Let $H \in \graphs$ be a graph and let $\tauli(H)=(V(H),E_{\mathsf{Li}}(H))$ be a graphical restriction defined as follows: $E_\mathsf{Li}(H) = \bs{ \bs{u,w} ~|~u \neq w \wedge \exists v: \bs{u,v},\bs{w,v} \in E(H)}$. Then for all $G \in \graphs$ it holds that $\taulihoms(H,G) = \lihoms(H,G)$.
\end{lemma}  
\begin{proof}
We prove both inclusions. Let $\varphi \in \taulihoms(H,G)$ and assume that $\varphi$ is not locally injective. Then there exists $v \in V(H)$ such that $\varphi|_{N(v)}$ is not injective which implies that there are $u$ and $w$ such that $\bs{u,v}$ and $\bs{w,v}$ are edges in $H$ and $\varphi(u)=\varphi(w)$. By definition $\bs{u,w} \in E_\mathsf{Li}(H)=E(\tauli(H))$ and therefore $\varphi \notin \taulihoms(H,G)$ which is a contradiction.\\
Now let $\varphi \in \lihoms(H,G)$ and assume that $\varphi \notin \taulihoms(H,G)$. Then there exist $u,w~\in~V(H)$ such that $\bs{u,w} \in E(\tauli(H))$ and $\varphi(u) = \varphi(w)$. The former implies that $u$ and $w$ have a common neighbor $v$ in $H$ but this contradicts the fact that $\varphi$ is locally injective.
\end{proof}

We continue by stating the dichotomy for counting locally injective homomorphisms.
\begin{corollary}[Corollary~\ref{cor:int_li_dicho}, restated]
\label{cor:dichotomy_li}
Let $C \subseteq \graphs$ be a recursively enumerable class of graphs.\\ Then $\#\lihoms(C)$ is FPT if $\taulimins(C)$ has bounded treewidth and $\#\W$-hard otherwise. Furthermore, there exists a deterministic algorithm that computes $\#\lihoms(H,G)$ in time  
\[ g(|V(H)|) \cdot |V(G)|^{\mathsf{tw}(\taulimins(H)) + 1} \,, \] where $g$ is a computable function.
\end{corollary}
\begin{proof}
Follows immediately from Lemma~\ref{lem:li_graphic} and Theorem~\ref{thm:graphic_dichotomy}.
\end{proof}

We give an example for a hard instance of the problem: Let $W_k$ be the ``windmill'' graph of size $k$, i.e., the graph with vertices $a$, $v_1,\dots,v_k$, $w_1,\dots,w_k$ and edges $\bs{a,v_i}, \bs{v_i,w_i}$ and $\bs{w_i,a}$ for each $i \in [k]$. Furthermore we let $\mathcal{W}$ be the set of all $W_k$ for $k \in \mathbb{N}$.

\begin{corollary}
\label{cor:lihoms_windmills_hard}
	$\#\lihoms(\mathcal{W})$ is $\#\W$-hard.
\end{corollary}
\begin{proof}
	It turns out that \emph{every} graph consisting of $k$ edges is a minor of some graph in $\taulimins(W_k)$. To see this let $F$ be a graph with $k$ edges. We enumerate the edges of $F$ as $e_1,\dots,e_k$ and identify each edge $e_i=\bs{x_i,y_i}$ with the edge $\bs{v_i,w_i}$ in $W_k$. Now, whenever $x_i = x_j$ (or $x_i = y_j$) we contract vertices $v_i$ and $v_j$ (or $v_i$ and $w_j$, respectively) in $W_k$. As each $v_i$ and $v_j$ (or $v_i$ and $w_j$, respectively) have the common neighbor $a$, and furthermore $v_i$ and $w_i$ are never contracted, the resulting graph $W'_k$ is a $\tauli$-minor of $W_k$. If we now remove $a$ from $W'_k$ along with every edge incident to $a$, the resulting graph is isomorphic to $F$. Consequently, the treewidth of $\taulimins(\mathcal{W})$ is not bounded and hence $\#\lihoms(\mathcal{W})$ is $\#\W$-hard by Theorem~\ref{thm:graphic_dichotomy}.
\end{proof}

In contrast to embeddings where every FPT case is also polynomial time solvable, there are ``real'' FPT cases when it comes to locally injective homomorphisms. Let $\mathcal{T} \subseteq \graphs$ be the class of all trees. Counting locally injective homomorphisms from those graphs is fixed-parameter tractable:
\begin{corollary}
	\label{col:lihoms_tree_fpt}
$\#\lihoms(\mathcal{T})$ is FPT. In particular, there is a deterministic algorithm that computes $\#\lihoms(T,G)$ for a tree $T$ in time
\[ g(|V(T)|) \cdot |V(G)|^2 \,, \] where $g$ is a computable function.
\end{corollary}
\begin{proof}
According to Corollary~\ref{cor:dichotomy_li} we only need to show that $\taumins(\mathcal{T})$ has treewidth $1$. Indeed, every $\tauli$-minor of a tree is again a tree, and has therefore treewidth $1$. To see this, consider a pair of vertices $u$ and $w$ that have a common neighbor $v$ in a tree $T \in \mathcal{T}$. Then $(u,v,w)$ is the only path between $u$ and $w$ and consequently contracting $u$ and $w$ to a single vertex will not create a cycle in the resulting graph (recall that we delete multiedges).
\end{proof}
On the other hand $\#\lihoms(\mathcal{T})$ is unlikely to have a polynomial time algorithm.
\begin{lemma}
\label{lem:li_trees_hard}
$\#\lihoms(\mathcal{T})$ is $\sharpP$-hard.
\end{lemma}
We prove this lemma in the following subsection.
\subsection{Counting subtrees of trees}
The aim of this section is to prove Lemma~\ref{lem:li_trees_hard}. We start by giving an introduction to classical counting complexity which was established by Valiant in his seminal work about the complexity of computing the permament \cite{Valiant1979a}. A (non-parameterized) \emph{counting problem} is a function $F: \bs{0,1}^\ast \rightarrow \N$. The class of all counting problems solvable in polynomial time is called $\FP$. On the other hand, the notion of intractability is $\sharpP$-hardness. $\sharpP$ is the class of all counting problems reducible\footnote{(Many-one) reductions in counting complexity differ slightly from many-one reductions in the decision world. However, for the purpose of this section we only need Turing reductions. We recommend Chap.~6.2 of \cite{goldreich} to the interested reader.} to $\#\mathsf{SAT}$, the problem of computing the number of satisfying assignments of a given CNF formula. A counting problem $F$ is $\sharpP$-hard if there exists a polynomial time Turing reduction from $\#\mathsf{SAT}$ to $F$, that is, an algorithm with oracle $F$ that solves $\#\mathsf{SAT}$ in polynomial time. Toda \cite{toda} proved that $\cc{PH}\subseteq \P^{\sharpP}$ which indicates that $\sharpP$-hard problems are much harder than $\NP$-complete problems.\\
To prove Theorem~\ref{lem:li_trees_hard}, we will first prove $\sharpP$-hardness of the following intermediate problem: Given two trees $T_1,T_2$, compute the number $\#\mathsf{Sub}(T_1,T_2)$ of subtrees of $T_2$ that are isomorphic to $T_1$. We call this problem $\subtt$. 
\begin{lemma}[Lemma~\ref{lem:into_classical_trees}, restated]
\label{lem:subtt_hard}
$\subtt$ is $\sharpP$-hard.
\end{lemma}
Related results are $\sharpP$-hardness for counting \emph{all} subtrees of a given graph \cite{jerrum_trees} or even counting \emph{all} subtrees of a given tree \cite{goldberg_jerrum_trees}. As the number of non-isomorphic trees with $n$ vertices is not bounded by a polynomial in $n$, we do not know how to reduce directly from these problems. Instead we use a construction quite similar to the ''skeleton'' graph in \cite{goldberg_jerrum_trees} to reduce from the problem of computing the permanent. \\
Given a quadratic matrix $A$ with elements $(a_{i,j})_{i,j \in [n]}$ the \emph{permanent} of $A$ is defined as follows:
\begin{equation*}
\mathsf{perm}(A)=\sum_{\pi \in S_n}\prod_{i=1}^n a_{i,\pi(i)} \,,
\end{equation*}  
where $S_n$ is the symmetric group with $n$ elements. 
\begin{theorem}[\cite{Valiant1979a}]
Computing the permanent is $\sharpP$-hard even when restricted to matrices with entries from $\bs{0,1}$.
\end{theorem}
\begin{proof}[Proof of Lemma~\ref{lem:subtt_hard}]
We reduce from computing the permanent of matrices with entries from $\bs{0,1}$. Given a quadratic matrix $A$ of size $n$, we construct a tree $T_A$ as follows:
\begin{figure}
	\begin{center}
		\begin{tikzpicture}[-]
  \node[circle,inner sep=2pt,fill](0) at (2,2) {};
  \node[circle,inner sep=2pt,fill](11) at (0,0) {};
  \node[circle,inner sep=2pt,fill](110) at (-0.4,-0.4) {};
  \node[circle,inner sep=2pt,fill](12) at (1,0) {};
  \node[circle,inner sep=2pt,fill](13) at (2,0) {};
  \node[circle,inner sep=2pt,fill](14) at (3,0) {};
  \node[circle,inner sep=2pt,fill](15) at (4,0) {};
  \node[circle,inner sep=2pt,fill](111) at (-0.3,-5.5) {};
  \node[circle,inner sep=2pt,fill](112) at (0,-5.5) {};
  \node[circle,inner sep=2pt,fill](113) at (0.3,-5.5) {};
  \node[circle,inner sep=2pt,fill](21) at (0,-1) {};
  \node[circle,inner sep=2pt,fill](22) at (1,-1) {};
  \node[circle,inner sep=2pt,fill](220) at (0.6,-1.4) {};
  \node[circle,inner sep=2pt,fill](23) at (2,-1) {};
  \node[circle,inner sep=2pt,fill](24) at (3,-1) {};
  \node[circle,inner sep=2pt,fill](25) at (4,-1) {};
  \node[circle,inner sep=2pt,fill](211) at (0.7,-5.5) {};
  \node[circle,inner sep=2pt,fill](212) at (1,-5.5) {};
  \node[circle,inner sep=2pt,fill](213) at (1.3,-5.5) {};
  \node[circle,inner sep=2pt,fill](31) at (0,-2) {};
  \node[circle,inner sep=2pt,fill](32) at (1,-2) {};
  \node[circle,inner sep=2pt,fill](33) at (2,-2) {};
  \node[circle,inner sep=2pt,fill](330) at (1.6,-2.4) {};
  \node[circle,inner sep=2pt,fill](34) at (3,-2) {};
  \node[circle,inner sep=2pt,fill](35) at (4,-2) {};
  \node[circle,inner sep=2pt,fill](311) at (1.7,-5.5) {};
  \node[circle,inner sep=2pt,fill](312) at (2,-5.5) {};
  \node[circle,inner sep=2pt,fill](313) at (2.3,-5.5) {};
  \node[circle,inner sep=2pt,fill](41) at (0,-3) {};
  \node[circle,inner sep=2pt,fill](42) at (1,-3) {};
  \node[circle,inner sep=2pt,fill](43) at (2,-3) {};
  \node[circle,inner sep=2pt,fill](44) at (3,-3) {};
  \node[circle,inner sep=2pt,fill](440) at (2.6,-3.4) {};
  \node[circle,inner sep=2pt,fill](45) at (4,-3) {};
  \node[circle,inner sep=2pt,fill](411) at (2.7,-5.5) {};
  \node[circle,inner sep=2pt,fill](412) at (3,-5.5) {};
  \node[circle,inner sep=2pt,fill](413) at (3.3,-5.5) {};
  \node[circle,inner sep=2pt,fill](51) at (0,-4) {};
  \node[circle,inner sep=2pt,fill](52) at (1,-4) {};
  \node[circle,inner sep=2pt,fill](53) at (2,-4) {};
  \node[circle,inner sep=2pt,fill](54) at (3,-4) {};
  \node[circle,inner sep=2pt,fill](55) at (4,-4) {};
  \node[circle,inner sep=2pt,fill](550) at (3.6,-4.4) {};
  \node[circle,inner sep=2pt,fill](511) at (3.7,-5.5) {};
  \node[circle,inner sep=2pt,fill](512) at (4,-5.5) {};
  \node[circle,inner sep=2pt,fill](513) at (4.3,-5.5) {};
  \node[circle,inner sep=2pt,fill](61) at (0,-5) {};
  \node[circle,inner sep=2pt,fill](62) at (1,-5) {};
  \node[circle,inner sep=2pt,fill](63) at (2,-5) {};
  \node[circle,inner sep=2pt,fill](64) at (3,-5) {};
  \node[circle,inner sep=2pt,fill](65) at (4,-5) {};
  \node[circle,inner sep=2pt,fill](71) at (0,1) {};
  \node[circle,inner sep=2pt,fill](72) at (1,1) {};
  \node[circle,inner sep=2pt,fill](73) at (2,1) {};
  \node[circle,inner sep=2pt,fill](74) at (3,1) {};
  \node[circle,inner sep=2pt,fill](75) at (4,1) {};
  \draw (11) -- (21);
  \draw (21) -- (31);
  \draw (12) -- (22);
  \draw (22) -- (32);
  \draw (13) -- (23);
  \draw (23) -- (33);
  \draw (31) -- (41);
  \draw (41) -- (51);
  \draw (51) -- (61);
  \draw (32) -- (42);
  \draw (52) -- (42);
  \draw (52) -- (62);
  \draw (33) -- (43);
  \draw (53) -- (43);
  \draw (53) -- (63);
  \draw (11) -- (71);
  \draw (12) -- (72);
  \draw (13) -- (73); 
  \draw (14) -- (24);
  \draw (24) -- (34);
  \draw (34) -- (44);
  \draw (44) -- (54);
  \draw (54) -- (64);
  \draw (14) -- (74);
  \draw (15) -- (25);
  \draw (25) -- (35);
  \draw (35) -- (45);
  \draw (45) -- (55);
  \draw (55) -- (65);
  \draw (15) -- (75);
  
  \draw (11) -- (110);
  \draw (22) -- (220);
  \draw (33) -- (330);
  \draw (44) -- (440);
  \draw (55) -- (550);
  
  \draw (61) -- (111);
  \draw (61) -- (112);
  \draw (61) -- (113);
  \draw (62) -- (211);
  \draw (62) -- (212);
  \draw (62) -- (213);
  \draw (63) -- (311);
  \draw (63) -- (312);
  \draw (63) -- (313);
  \draw (64) -- (411);
  \draw (64) -- (412);
  \draw (64) -- (413);
  \draw (65) -- (511);
  \draw (65) -- (512);
  \draw (65) -- (513);
  
  \draw (0) -- (71);
  \draw (0) -- (72);
  \draw (0) -- (73);
  \draw (0) -- (74);
  \draw (0) -- (75);
\end{tikzpicture} \hspace*{20mm} \begin{tikzpicture}[-]
  \node[circle,inner sep=2pt,fill](0) at (2,2) {};
  \node[circle,inner sep=2pt,fill](11) at (0,0) {};
  \node[circle,inner sep=2pt,fill](110) at (-0.4,-0.4) {};
  \node[circle,inner sep=2pt,fill](12) at (1,0) {};
  \node[circle,inner sep=2pt,fill](13) at (2,0) {};
  \node[circle,inner sep=2pt,fill](130) at (1.6,-0.4) {};
  \node[circle,inner sep=2pt,fill](14) at (3,0) {};
  \node[circle,inner sep=2pt,fill](15) at (4,0) {};
  \node[circle,inner sep=2pt,fill](111) at (-0.3,-5.5) {};
  \node[circle,inner sep=2pt,fill](112) at (0,-5.5) {};
  \node[circle,inner sep=2pt,fill](113) at (0.3,-5.5) {};
  \node[circle,inner sep=2pt,fill](21) at (0,-1) {};
  \node[circle,inner sep=2pt,fill](22) at (1,-1) {};
  \node[circle,inner sep=2pt,fill](220) at (0.6,-1.4) {};
  \node[circle,inner sep=2pt,fill](23) at (2,-1) {};
  \node[circle,inner sep=2pt,fill](24) at (3,-1) {};
  \node[circle,inner sep=2pt,fill](240) at (2.6,-1.4) {};
  \node[circle,inner sep=2pt,fill](25) at (4,-1) {};
  \node[circle,inner sep=2pt,fill](250) at (3.6,-1.4) {};
  \node[circle,inner sep=2pt,fill](211) at (0.7,-5.5) {};
  \node[circle,inner sep=2pt,fill](212) at (1,-5.5) {};
  \node[circle,inner sep=2pt,fill](213) at (1.3,-5.5) {};
  \node[circle,inner sep=2pt,fill](31) at (0,-2) {};
  \node[circle,inner sep=2pt,fill](310) at (-0.4,-2.4) {};
  \node[circle,inner sep=2pt,fill](32) at (1,-2) {};
  \node[circle,inner sep=2pt,fill](33) at (2,-2) {};
  \node[circle,inner sep=2pt,fill](330) at (1.6,-2.4) {};
  \node[circle,inner sep=2pt,fill](34) at (3,-2) {};
  \node[circle,inner sep=2pt,fill](35) at (4,-2) {};
  \node[circle,inner sep=2pt,fill](311) at (1.7,-5.5) {};
  \node[circle,inner sep=2pt,fill](312) at (2,-5.5) {};
  \node[circle,inner sep=2pt,fill](313) at (2.3,-5.5) {};
  \node[circle,inner sep=2pt,fill](41) at (0,-3) {};
  \node[circle,inner sep=2pt,fill](42) at (1,-3) {};
  \node[circle,inner sep=2pt,fill](420) at (0.6,-3.4) {};
  \node[circle,inner sep=2pt,fill](43) at (2,-3) {};
  \node[circle,inner sep=2pt,fill](44) at (3,-3) {};
  \node[circle,inner sep=2pt,fill](440) at (2.6,-3.4) {};
  \node[circle,inner sep=2pt,fill](45) at (4,-3) {};
  \node[circle,inner sep=2pt,fill](411) at (2.7,-5.5) {};
  \node[circle,inner sep=2pt,fill](412) at (3,-5.5) {};
  \node[circle,inner sep=2pt,fill](413) at (3.3,-5.5) {};
  \node[circle,inner sep=2pt,fill](51) at (0,-4) {};
  \node[circle,inner sep=2pt,fill](52) at (1,-4) {};
  \node[circle,inner sep=2pt,fill](520) at (0.6,-4.4) {};
  \node[circle,inner sep=2pt,fill](53) at (2,-4) {};
  \node[circle,inner sep=2pt,fill](54) at (3,-4) {};
  \node[circle,inner sep=2pt,fill](55) at (4,-4) {};
  \node[circle,inner sep=2pt,fill](550) at (3.6,-4.4) {};
  \node[circle,inner sep=2pt,fill](511) at (3.7,-5.5) {};
  \node[circle,inner sep=2pt,fill](512) at (4,-5.5) {};
  \node[circle,inner sep=2pt,fill](513) at (4.3,-5.5) {};
  \node[circle,inner sep=2pt,fill](61) at (0,-5) {};
  \node[circle,inner sep=2pt,fill](62) at (1,-5) {};
  \node[circle,inner sep=2pt,fill](63) at (2,-5) {};
  \node[circle,inner sep=2pt,fill](64) at (3,-5) {};
  \node[circle,inner sep=2pt,fill](65) at (4,-5) {};
  \node[circle,inner sep=2pt,fill](71) at (0,1) {};
  \node[circle,inner sep=2pt,fill](72) at (1,1) {};
  \node[circle,inner sep=2pt,fill](73) at (2,1) {};
  \node[circle,inner sep=2pt,fill](74) at (3,1) {};
  \node[circle,inner sep=2pt,fill](75) at (4,1) {};
  \draw (11) -- (21);
  \draw (21) -- (31);
  \draw (12) -- (22);
  \draw (22) -- (32);
  \draw (13) -- (23);
  \draw (23) -- (33);
  \draw (31) -- (41);
  \draw (41) -- (51);
  \draw (51) -- (61);
  \draw (32) -- (42);
  \draw (52) -- (42);
  \draw (52) -- (62);
  \draw (33) -- (43);
  \draw (53) -- (43);
  \draw (53) -- (63);
  \draw (11) -- (71);
  \draw (12) -- (72);
  \draw (13) -- (73); 
  \draw (14) -- (24);
  \draw (24) -- (34);
  \draw (34) -- (44);
  \draw (44) -- (54);
  \draw (54) -- (64);
  \draw (14) -- (74);
  \draw (15) -- (25);
  \draw (25) -- (35);
  \draw (35) -- (45);
  \draw (45) -- (55);
  \draw (55) -- (65);
  \draw (15) -- (75);
  
  \draw (11) -- (110);
  \draw (31) -- (310);
  \draw (22) -- (220);
  \draw (13) -- (130);
  \draw (33) -- (330);
  \draw (24) -- (240);
  \draw (25) -- (250);
  \draw (42) -- (420);
  \draw (44) -- (440);
  \draw (52) -- (520);
  \draw (55) -- (550);
  
  \draw (61) -- (111);
  \draw (61) -- (112);
  \draw (61) -- (113);
  \draw (62) -- (211);
  \draw (62) -- (212);
  \draw (62) -- (213);
  \draw (63) -- (311);
  \draw (63) -- (312);
  \draw (63) -- (313);
  \draw (64) -- (411);
  \draw (64) -- (412);
  \draw (64) -- (413);
  \draw (65) -- (511);
  \draw (65) -- (512);
  \draw (65) -- (513);
  
  \draw (0) -- (71);
  \draw (0) -- (72);
  \draw (0) -- (73);
  \draw (0) -- (74);
  \draw (0) -- (75);
\end{tikzpicture}
	\end{center}
	\caption{Trees $T_{\mathsf{id}_5}$ (left) and $T_A$ (right).}
	\label{fig:matrix_tree}
\end{figure}
\begin{enumerate}
\item For every entry $a_{i,j}$ we create a vertex $v_{i,j}$ and add edges $\bs{v_{i,j},v_{i+1,j}}$ for every $i \in [n-1]$ and $j \in [n]$.
\item Whenever $a_{i,j}=1$ we create a vertex $b_{i,j}$ and add edges $\bs{b_{i,j},v_{i,j}}$.
\item For every column $c_j$ we create a vertices $u_j,w_j,x_j,y_j,z_j$ and add edges $\bs{u_j,v_{1,j}}$, $\bs{v_{n,j},w_j}$,$\bs{w_j,x_j}$,$\bs{w_j,y_j}$ and $\bs{w_j,z_j}$. 
\item Finally, we create a vertex $r$ and add edges $\bs{a,u_j}$ for all $j\in [n]$. In the following we call $r$ the root.
\end{enumerate}
We give an example in Figure~\ref{fig:matrix_tree} for a matrix \[A =
		\begin{pmatrix}
		1 & 0 & 1 & 0 & 0 \\
		0 & 1 & 0 & 1 & 1 \\
		1 & 0 & 1 & 0 & 0\\
		0 & 1 & 0 & 1 & 0\\
		0 & 1 & 0 & 0 & 1
		\end{pmatrix}
		\,.\]
We claim that for all quadratic matrices $A$ of size $n\geq 5$ with entries from $\bs{0,1}$ it holds that
\begin{equation*}
\mathsf{perm}(A) = \#\mathsf{Sub}(T_{\mathsf{id}_n},T_A) \,,
\end{equation*}
where $\mathsf{id}_n$ is the quadratic matrix of size $n$ with $1$s on the diagonal and $0$s everywhere else. In the following we write $v$ for a vertex in $T_A$ and $v'$ for a vertex in $T_{\mathsf{id}_n}$. To prove the claim we first observe that whenever a subtree of $T_A$ is isomorphic to $T_{\mathsf{id}_n}$, the root $r'$ of $T_{\mathsf{id}_n}$ has to be mapped to the root $r$ of $T_A$ by the isomorphism as the roots are the only vertices with degree $n$ (which is why we needed $n\geq 5$ as every other vertex has degree~$\leq 4$). It follows that the vertices $u'_1,\dots,u'_n$ of $T_{\mathsf{id}_n}$ are mapped to $u_1,\dots,u_n$ of $T_A$ which induces a permutation on $n$ elements, that is, an element $\pi \in S_n$. We will now partition the subtrees of $T_A$ isomorphic to $T_{\mathsf{id}_n}$ by those permutations and write $\#\mathsf{Sub}(T_{\mathsf{id}_n},T_A)[\pi]$ for the number of subtrees that induce $\pi$. Now fix $\pi$ and consider a subtree that induces $\pi$. It holds that for all $j\in [n]$ the vertex $w'_j$ has to be mapped to $w_{\pi(j)}$ as those are the only vertices with degree exactly $4$ and furthermore, the vertices $x'_j,y'_j,z'_j$ have to be mapped to $x_{\pi(j)},y_{\pi(j)},z_{\pi(j)}$ (possibly permuted but the subtree of $T_A$ is the same). Now $v'_{i,i}$ is adjacent to $b'_{i,i}$ for each $i\in [n]$ and therefore $v_{i,\pi(i)}$ has to be adjacent to $b_{i,\pi(i)}$, that is $a_{i,\pi(i)} = 1$. If this is not the case then there is no subtree that induces partition $\pi$. Furthermore there is at most one subtree isomorphic to $T_{\mathsf{id}_n}$ inducing $\pi$ because the image is enforced by $r'$, $w'_j$ and $v'_{i,i}$ for all $i,j \in [n]$. Consequently $\#\mathsf{Sub}(T_{\mathsf{id}_n},T_A)[\pi] = 1$ if for all $i \in [n]$ it holds that $a_{i,\pi(i)}=1$ and $\#\mathsf{Sub}(T_{\mathsf{id}_n},T_A)[\pi]=0$ otherwise. Hence $\#\mathsf{Sub}(T_{\mathsf{id}_n},T_A)[\pi] = \prod_{i=1}^n a_{i,\pi(i)}$ and therefore
\begin{equation*}
\mathsf{perm}(A) = \sum_{\pi \in S_n}\prod_{i=1}^n a_{i,\pi(i)} = \sum_{\pi \in S_n}\#\mathsf{Sub}(T_{\mathsf{id}_n},T_A)[\pi] = \#\mathsf{Sub}(T_{\mathsf{id}_n},T_A) \,.
\end{equation*}
Now the reduction works as follows: If the input matrix $A$ has size $\leq 4$ we brute-force the output and otherwise we compute $\#\mathsf{Sub}(T_{\mathsf{id}_n},T_A)$ with the oracle for $\subtt$.
\end{proof}

Now the proof of Lemma~\ref{lem:li_trees_hard} relies on the fact that locally injective homomorphisms from a tree to a tree are embeddings.
\begin{proof}[Proof of Lemma~\ref{lem:li_trees_hard}]
It is a well-known fact that a locally injective homomorphism $\varphi$ from a tree $T_1$ to a tree $T_2$ is injective. To see this assume that there are vertices $v$ and $u$ in $T_1$ that are mapped to the same vertex in $T_2$. As $T_1$ is a tree there exists exactly one path $v=w_0,w_1,\dots,w_\ell,w_{\ell+1}=u$ between $v$ and $u$ in $T_1$. It holds that $\ell \geq 1$ as otherwise $v$ and $u$ would be adjacent and hence $\varphi(u)=\varphi(v)$ would have a selfloop in $T_2$ which is impossible. As $\varphi$ is locally injective we have that $\varphi(v)\neq \varphi(w_2)$, hence $u \neq w_2$, and as $\varphi$ is edge preserving there are edges $\bs{\varphi(v),\varphi(w_1)}$ and $\bs{\varphi(w_1),\varphi(w_2)}$ and a path from $\varphi(w_2)$ to $\varphi(w_{\ell+1})=\varphi(u)=\varphi(v)$ in $T_2$. This induces a cycle and contradicts the fact that $T_2$ is a tree.\\
Therefore $\#\embs(T_1,T_2) = \#\lihoms(T_1,T_2)$. By Lov{\'a}sz \cite{lovasz} it holds for all $H$ and $G$ that \[\#\mathsf{Sub}(H,G) = \frac{\#\embs(H,G)}{\#\mathsf{Aut}(H)} \,,\]
where $\mathsf{Aut}(H)$ is the set of automorphisms of $H$. If $H$ is a tree then $\#\mathsf{Aut}(H) $ can be computed in polynomial time (even for planar graphs \cite{isoauto},\cite{isoplanar}). Therefore $\sharpP$-hardness of $\#\lihoms(\mathcal{T})$ follows by reducing from $\subtt$: Given trees $T_1,T_2$ we compute $\#\lihoms(T_1,T_2)$ by querying the oracle and $\#\mathsf{Aut}(T_1)$ in polynomial time. Then we output
\[\frac{\#\lihoms(T_1,T_2)}{\#\mathsf{Aut}(T_1)} = \frac{\#\embs(T_1,T_2)}{\#\mathsf{Aut}(T_1)} = \#\mathsf{Sub}(T_1,T_2) \,. \]
\end{proof}

\section{Injectivity in r-neighborhood}
The generalization from locally injective homomorphisms to homomorphisms that are injective in the $r$-neighborhood of every vertex is straightforward. Given a graph $H$ and $v \in V(H)$ we denote $N_r(v)$ as the $r$-neighborhood of $v$, that is, a vertex $u$ is contained in $N_r(v)$ if and only if $d_H(u,v) \leq r$, where $d_H(u,v)$ the distance between $u$ and $v$ in $H$. We then define \[\lirhoms(H,G) := \bs{\varphi \in \homs(H,G)~|~\forall v \in V(H) ~:~\varphi|_{N_r(v)} \text{ is injective}} \,. \]
Furthermore we define the counting problem $\#\lirhoms(C)$ for a class of graphs $C$ accordingly. Defining $\taulir$ such that \[E(\taulir(H)) = \bs{ \bs{u,w} ~|~ u \neq w \wedge \exists v: 1 \leq d_H(u,v) \leq r \wedge  1 \leq d_H(w,v) \leq r} \] for every graph $H \in \graphs$ immediately yields the dichotomy:

\begin{corollary}
Let $C \subseteq \graphs$ be a recursively enumerable class of graphs. Then $\#\lirhoms(C)$ is FPT if $\taulirmins(C)$ has bounded treewidth and $\#\W$-hard otherwise. Furthermore, there exists a deterministic algorithm that computes $\#\lirhoms(H,G)$ in time  
\[ g(|V(H)|) \cdot |V(G)|^{\mathsf{tw}(\taulirmins(H)) + 1} \,, \] where $g$ is a computable function.
\end{corollary} 

We continue using trees as an example by observing that there is a phase transition in the complexity of $\#\lirhoms(\mathcal{T})$ when we change from $r=1$, in which case $\lirhoms(H,G)=\lihoms(H,G)$, to $r=2$:

\begin{corollary}
$\#\lirhoms(\mathcal{T})$ is $\#\W$-hard for $r \geq 2$. In particular, assuming ETH \footnote{ETH is the ``exponential time hypothesis'', stating that $k$-SAT cannot be solved in subexponential time (see \cite{ksat}).}, there is no algorithm that computes $\#\lirhoms(T,G)$ for a tree $T$ in time
\[ g(|V(T)|) \cdot |V(G)|^{O(1)} \,, \]
for any computable function $g$.
\end{corollary}

\begin{proof}
We only need to show that $\taulirmins(\mathcal{T})$ has unbounded treewidth, as the ETH lower bound simply follows from the fact that FPT $\neq \#\W$ under ETH (see e.g. Chapt. 16 in \cite{flumgrohe_counting}). Therefore we construct the graph $T_{k,k}$ as follows: 
\begin{itemize}
\item We add vertices $a$, $u_1,\dots,u_{k}$, $v_1,\dots,v_{k}$, $w_1,\dots,w_{k}$.
\item We add edges $\bs{a,u_i}$, $\bs{a,v_i}$ and $\bs{v_i,w_i}$ for $i \in [k]$.
\end{itemize}
Clearly, $T_{k,k}$ is a tree. Now we contract vertices $w_i$ and $u_i$ for all $i \in [k]$ and end up in $W_k$. As $d_{T_{k,k}}(a,w_i) = 2$ and $d_{T_{k,k}}(a,u_i)=1$, those contractions are according to $\taulir(T_{k,k})$ and hence the resulting graph is a $\taulir$-minor of $T_{k,k}$. From $W_k$ we can further contract vertices along the lines of the proof of Corollary~\ref{cor:lihoms_windmills_hard} to obtain arbitrary graphs with $k$ edges as minors of elements of $\taulirmins(T_{k,k})$. Consequently the treewidth of $\taulir(\mathcal{T})$ is not bounded.
\end{proof}

\section{Extension to linear combinations}

\label{sec:lincombs}

The introduction of linear combinations of graphically restricted homomorphisms is motivated by the following example: Consider the problem $\#\mathsf{E}$ of, given a parameter $k$ and a graph $G \in \graphs$, computing $\#\homs(P_k,G) + \#\lihoms(K_k,G) + \#\embs(C_k,G)$, where $P_k,K_k,C_k$ are paths, cliques and cycles consisting of $k$ vertices. As $\homs$\footnote{Here $\tau$ maps every graph to the independent set of the same size implying that $\taumins(C)=C$.}, $\lihoms$ and $\embs$ are graphically restricted homomorphisms we know the complexity of computing each summand, but we cannot immediately infer the complexity of $\#\mathsf{E}$. As $P_k$ has treewidth $1$ it follows by Theorem~\ref{thm:hom-dichotomy} or Theorem~\ref{thm:graphic_dichotomy} that $\#\homs(P_k,G)$ can be computed in FPT time. Consequently, $\#\mathsf{E}$ is equivalent (w.r.t. FPT Turing reductions) to computing $\#\lihoms(K_k,G) + \#\embs(C_k,G)$. As cliques have $\tauli$-minors of unbounded treewidth and cycles have unbounded matching number, these problems are both $\#\W$-hard (see Theorem~\ref{cor:dichotomy_li} and Theorem~\ref{thm:emb-dichotomy}). Even if hardness of $\#\mathsf{E}$ is intuitive, it is not obvious how to prove it, at least if one tries to reduce the computation of one summand to $\#\mathsf{E}$. Instead we will show that our framework allows less cumbersome reductions, at least for what we will call the \emph{congruent} cases. We start by formally defining a linear combination of graphically restricted homomorphisms.
\begin{definition}
\label{def:lincombs_graphic}
Let $\mathcal{A}$ be a set of computable functions $a:~\mathcal{G} \times \Tau~\rightarrow~\Q_{\geq 0}$ with finite support. 
We define the parameterized counting problem $\#\lc(\mathcal{A})$ as follows: Given $a \in \mathcal{A}$ and $G \in \mathcal{G}$, compute the \emph{linear combination}:
\[ \sum_{(H,\tau) \in \mathsf{supp}(a)} a(H,\tau) \cdot \#\tauhoms(H, G) \,, ~\text{parameterized by } \left(\#\mathsf{supp}(a)+\max_{(H,\tau) \in \mathrm{supp}(a)} \#V(H)\right) \,.\] 
Given a function $a \in \mathcal{A}$ we denote \[\Taumins(a):= \bigcup_{(H,\tau)\in \mathsf{supp}(a)}\taumins(H)~~~ \text{ and }~~~ \Taumins(\mathcal{A}) := \bigcup_{a\in \mathcal{A}} \Taumins(a)\] as the set of all $\tau$-minors of $a$ and $\mathcal{A}$, respectively. Furthermore, we say that $a$ is \emph{congruent} if for every $(H_1,\ast)$ and $(H_2,\ast) \in \mathsf{supp}(a)$ it holds that $\mathsf{Parity}(\#V(H_1)) = \mathsf{Parity}(\#V(H_2))$. We say that $\mathcal{A}$ is congruent if all its elements are congruent. 
\end{definition}

If we let $\tau_\mathsf{is}$ be the graphical restriction that maps a graph $H$ to the independent set with vertices $V(H)$ and set $\mathcal{A}=\bs{a_k ~|~k\in \N}$ such that $a_k(P_k,\tau_\mathsf{is})=1$, $a_k(K_k,\tauli)=1$, $a_k(C_k,\tau_\mathsf{clique}) = 1$ and $0$ otherwise then $\#\lc(\mathcal{A})$ is equivalent to $\#\mathsf{E}$.\\

For congruent $\mathcal{A}$ we can derive a complete complexity classification. 
\begin{theorem}
\label{thm:graphic_lincombs_dichotomy}
The problem $\#\lc(\mathcal{A})$ is fixed-parameter tractable if $\Taumins(\mathcal{A})$ has bounded treewidth. Otherwise, if $\mathcal{A}$ is additionally congruent, it is $\#\W$-hard.
\end{theorem}
\begin{proof}
The FPT algorithm for the positive result is straight-forward: As the treewidth of $\Taumins(\mathcal{A})$ is bounded, we can on input $a \in \mathcal{A}$ and $G\in \graphs$ compute $\#\tauhoms(H,G)$ for every $(H,\tau)\in\mathsf{supp}(a)$ in time $g(\#V(H)) \cdot n^{O(1)}$ for a computable function $g$ by Theorem~\ref{thm:graphic_dichotomy}. Consequently, computing the sum takes time less than \[\#\mathsf{supp}(a) \cdot g\left(\max_{(H,\tau) \in \mathrm{supp}(a)} \#V(H)\right) \cdot n^{O(1)} \]
yielding fixed-parameter tractability.\\
Now assume that $\Taumins(\mathcal{A})$ has unbounded treewidth and that $\mathcal{A}$ is congruent and let $a\in \mathcal{A}$ and $G\in \graphs$. Lemma~\ref{lem:main_graphic} yields that
\begin{align*}
\label{eqn:lincombs_graphic}
&\sum_{(H,\tau) \in \mathsf{supp}(a)} a(H,\tau) \cdot \#\tauhoms(H, G)\\ 
&= \sum_{(H,\tau) \in \mathsf{supp}(a)} a(H,\tau) \cdot \sum_{\rho \geq \emptyset} \mu(\emptyset,\rho)\cdot \#\homs(H/\rho,G)\\
&= \sum_{(H,\tau) \in \mathsf{supp}(a)} ~\sum_{\rho \geq \emptyset} a(H,\tau) \cdot \mu(\emptyset,\rho)\cdot \#\homs(H/\rho,G)
\end{align*}
Now let $\mathcal{H} \in \Taumins(\mathcal{A})$. It holds that the coefficient $d(\mathcal{H})$ of $\#\homs(\mathcal{H},G)$ in the above equation satisfies:
\begin{equation*}
d(\mathcal{H}) = \sum_{(H,\tau) \in \mathsf{supp}(a)} ~\sum_{\substack{\rho \geq \emptyset \\  \mathcal{H} \cong H/\rho}} a(H,\tau) \cdot \mu(\emptyset,\rho)
\end{equation*}
If we fix some $(H,\tau) \in \mathsf{supp}(a)$ and $\rho \in L(M(\tau(H)))$ such that $\mathcal{H} \cong H/\rho$ we have that 
\begin{align*}
&\mathsf{sgn}\left(a(H,\tau) \cdot \mu(\emptyset,\rho)  \right) = \mathsf{sgn}\left( \mu(\emptyset,\rho) \right)
= (-1)^{\mathsf{rk}(\rho)} = (-1)^{\#V(H)-c(H/\rho)}= (-1)^{\#V(H)-c(\mathcal{H})} \,,
\end{align*}
where the first equality follows from the fact that $a(H,\tau)>0$ and the second from the corollary of Rota's Theorem (Theorem~\ref{thm:rota}). As $a$ is congruent, the parities of all $H$ such that $(H,\ast)\in \mathsf{supp}(a)$ are equal and consequently we have that $\mathsf{sgn}(d(\mathcal{H})) = (-1)^{\#V(H)-c(\mathcal{H})}$, hence $d(\mathcal{H}) \neq 0$. Therefore, if we consider $\#\lc(\mathcal{A})$ as the problem of computing linear combinations of homomorphisms (as we also did in the proof of Theorem~\ref{thm:graphic_dichotomy}), we infer that every $\tau$-minor will be inluded in the combination. As the treewidth of those is not bounded we conclude by Theorem~\ref{thm:hom-dichotomy} that $\#\lc(\mathcal{A})$ is $\#\W$-hard.
\end{proof}
On the other hand, Theorem~\ref{thm:graphic_lincombs_dichotomy} is \emph{not} true if we omit the constraint that $\mathcal{A}$ is congruent: Consider the problem $\#\homs(\mathcal{P})$ where $\mathcal{P}$ is the class of all paths. It is fixed-parameter tractable as $\mathcal{P}$ has bounded treewidth (see Theorem~\ref{thm:hom-dichotomy}). Using Lov{\'a}sz identity \cite{lovasz} we have that for any $P_k \in \mathcal{P}$ and $G \in \graphs$ it holds that
\[ \#\homs(P_k,G) = \sum_{\rho \geq \emptyset}\#\embs(P_k/\rho,G) \,.\]
This is a linear combination of graphical homomorphisms (embeddings) including e.g. the term $\#\embs(P_k/\emptyset,G)=\#\embs(P_k,G)$ with coefficient $1$. But $\tau_\mathsf{clique}\text{-}\mathcal{M}(\mathcal{P})$ has unbounded treewidth\footnote{$\tau_\mathsf{clique}\text{-}\mathcal{M}(P_k)$ is precisely the set of ``spasms'' of $P_k$ (see \cite{hombasis2017}). The claim follows by Fact~3.4 in \cite{hombasis2017}} and consequently the treewidth of all $\tau$-minors of this linear combination is unbounded, too. This shows that there exist non-congruent $\mathcal{A}$ such that the treewidth of $\Taumins(\mathcal{A})$ is not bounded but $\#\lc(\mathcal{A})$ is fixed-parameter tractable.

Now it is easy to see that $\#\mathsf{E}$ is $\#\W$-hard. Further problems whose hardness follows from Theorem~\ref{thm:graphic_lincombs_dichotomy} are for example:

\begin{corollary}
\label{cor:hard_lincombs}
The following problems are $\#\W$-hard: Given a graph $G \in \graphs$ and a parameter $k$,
\begin{enumerate}
\item[(1)] count all odd (or even) subgraphs of size bounded by $k$ of $G$. 
\item[(2)] count all subgraphs of size $k$ of $G$ (follows also from \cite{hombasis2017}).
\item[(3)] compute $\sum_{i =1}^k \#\lihoms(W_i,G)$, i.e., the sum of all locally injective homomorphisms from windmills of size bounded by $k$ to $G$.
\item[(4)] compute $\sum_{i=1}^k \#\lihoms(K_{i,i},G) + \#\embs(K_{i,i},G)$, where $K_{i,i}$ is the biclique of size $i$, that is, the complete bipartite graph with $i$ vertices on each side.
\end{enumerate}
\end{corollary}
\begin{proof}[Proof of Corollary~\ref{cor:hard_lincombs}]
Each statement follows by Theorem~\ref{thm:graphic_lincombs_dichotomy}:
\begin{enumerate}
\item[(1)] Let $\mathsf{Odd}_k \subseteq \graphs$ be the set of all odd graphs of size bounded by $k$. Then it holds that
\begin{equation*}
\sum_{H \in \mathsf{Odd}_k} \#\mathsf{Sub}(H,G) = \sum_{H \in \mathsf{Odd}_k} \#\mathsf{Aut}(H)^{-1} \cdot \#\embs(H,G)\,.
\end{equation*}
As $\mathsf{Aut}(H)^{-1} > 0$, the above equation clearly is a congruent instance of the linear combination problem. Furthermore $\mathsf{Odd}_k$ contains cliques of size $\geq k-1$ implying that the treewidth of the instance is not bounded. The same argument holds for the case of counting all even subgraphs.
\item[(2)] Follows also along the same lines as (1) with the additional argument that we only count graphs of size exactly $k$, implying that the parity is the same for all terms.
\item[(3)] Congruence follows by the observation that $W_i$ has odd size for all $i \geq 1$. Unbounded treewidth follows with the same argument as in Corollary~\ref{cor:lihoms_windmills_hard}.
\item[(4)] Congruence follows by the fact that $K_{i,i}$ has even size for all $i \geq 1$. Unbounded treewidth follows by observing that the class of all bicliques itself already has unbounded treewidth.
\end{enumerate}
\end{proof}

\section{Conclusion and further work}
We have shown that various parameterized counting problems can be expressed as a linear combination of homomorphisms over the lattice of graphic matroids, implying immediate complexity classifications along with fixed-parameter tractable algorithms for the positive cases. This results can be obtained without using often cumbersome tools like ``gadgeting'' or interpolation and relies only on the knowledge of the problem of counting homomorphisms and the comprehension of the cancellation behaviour when transforming a problem into this ``homomorphism basis''. The latter, in turn, was nothing more than a question about the sign of the Möbius function, which was answered by Rota's Theorem.\\
This framework, however, still has limits: It seems that, e.g., neither induced subgraphs nor edge-injective homomorphisms \cite{edge-injective17} are graphically restricted. Indeed, both can be expressed as a sum of homomorphisms over (non-geometric) lattices but the problem is that there are isomorphic terms with different signs in both cases. This suggests that a better understanding of the Möbius function over those lattices could yield even more general complexity classifications of parameterized counting problems.

\section*{Acknowledgements}
The author is very grateful to Holger Dell and Radu Curticapean for fruitful discussions. Furthermore the author thanks Cornelius Brand for saying ``Tutte Polynomial'' every once in a while.

\bibliography{geo-homs-bib}{}

\begin{thebibliography}{10}

\bibitem{blass1997mobius}
Andreas Blass and Bruce~E Sagan.
\newblock M{\"o}bius functions of lattices.
\newblock {\em Advances in mathematics}, 127(1):94--123, 1997.

\bibitem{best}
Graham~R Brightwell and Peter Winkler.
\newblock Note on counting eulerian circuits.

\bibitem{bulatov}
Andrei~A Bulatov.
\newblock A dichotomy theorem for nonuniform {C}{S}{P}s.
\newblock {\em arXiv preprint arXiv:1703.03021}, 2017.

\bibitem{fibonaccigates}
Jin-Yi Cai, Pinyan Lu, and Mingji Xia.
\newblock Holographic algorithms by fibonacci gates and holographic reductions
  for hardness.
\newblock In {\em Proceedings of the 49th Annual Symposium on Foundations of
  Computer Science, FOCS}, pages 644--653. IEEE, 2008.

\bibitem{holantdicho}
Jin-Yi Cai, Pinyan Lu, and Mingji Xia.
\newblock Computational complexity of holant problems.
\newblock {\em SIAM Journal on Computing}, 40(4):1101--1132, 2011.

\bibitem{kmatchings}
Radu Curticapean.
\newblock Counting matchings of size k is {\#}{W}[1]-hard.
\newblock In {\em Proceedings of the 40th International Colloquium on Automata,
  Languages, and Programming, ICALP, Part I}, pages 352--363, Berlin,
  Heidelberg, 2013. Springer Berlin Heidelberg.

\bibitem{hombasis2017}
Radu Curticapean, Holger Dell, and D\'aniel Marx.
\newblock Homomorphisms are a good basis for counting small subgraphs.
\newblock In {\em Proceedings of the 49th ACM Symposium on Theory of Computing,
  STOC}, pages 210--223, 2017.

\bibitem{edge-injective17}
Radu Curticapean, Holger Dell, and Marc Roth.
\newblock Counting edge-injective homomorphisms and matchings on restricted
  graph classes.
\newblock In {\em Proceedings of the 34th Symposium on Theoretical Aspects of
  Computer Science, STACS}, pages 25:1--25:15, 2017.

\bibitem{embsdicho}
Radu Curticapean and D{\'{a}}niel Marx.
\newblock Complexity of counting subgraphs: Only the boundedness of the
  vertex-cover number counts.
\newblock In {\em Proceedings of the 55th Annual Symposium on Foundations of
  Computer Science, FOCS}, pages 130--139, 2014.

\bibitem{cygan2015parameterized}
Marek Cygan, Fedor~V Fomin, {\L}ukasz Kowalik, Daniel Lokshtanov, D{\'a}niel
  Marx, Marcin Pilipczuk, Micha{\l} Pilipczuk, and Saket Saurabh.
\newblock {\em Parameterized algorithms}, volume~3.
\newblock Springer.

\bibitem{homsdicho1}
V{\'\i}ctor Dalmau and Peter Jonsson.
\newblock The complexity of counting homomorphisms seen from the other side.
\newblock {\em Theoretical Computer Science}, 329(1):315--323, 2004.

\bibitem{appr2}
Martin Dyer, Leslie~Ann Goldberg, and Mark Jerrum.
\newblock An approximation trichotomy for {B}oolean {\#}{CSP}.
\newblock {\em Journal of Computer and System Sciences}, 76(3):267--277, 2010.

\bibitem{federvardi}
Tom{\'a}s Feder and Moshe~Y Vardi.
\newblock The computational structure of monotone monadic snp and constraint
  satisfaction: A study through datalog and group theory.
\newblock {\em SIAM Journal on Computing}, 28(1):57--104, 1998.

\bibitem{fiala}
Ji{\v{r}}{\'\i} Fiala and Jan Kratochv{\'\i}l.
\newblock Locally constrained graph homomorphisms—structure, complexity, and
  applications.
\newblock {\em Computer Science Review}, 2(2):97--111, 2008.

\bibitem{flumgrohe_counting}
J{\"{o}}rg Flum and Martin Grohe.
\newblock The parameterized complexity of counting problems.
\newblock {\em {SIAM} J. Comput.}, 33(4):892--922, 2004.

\bibitem{goldberg_jerrum_trees}
Leslie~Ann Goldberg and Mark Jerrum.
\newblock Counting unlabelled subtrees of a tree is\#{P}-complete.
\newblock 1999.

\bibitem{goldreich}
Oded Goldreich.
\newblock {\em Computational Complexity: A Conceptual Perspective}.
\newblock Cambridge University Press, 2008.

\bibitem{isoplanar}
John~E Hopcroft and Robert~Endre Tarjan.
\newblock Isomorphism of planar graphs.
\newblock In {\em Complexity of computer computations}, pages 131--152.
  Springer, 1972.

\bibitem{ksat}
Russell Impagliazzo and Ramamohan Paturi.
\newblock Complexity of k-{S}{A}{T}.
\newblock In {\em Proceedings of the 14th Annual IEEE Conference on
  Computational Complexity, 1999}, pages 237--240. IEEE, 1999.

\bibitem{jerrum_trees}
Mark Jerrum.
\newblock Counting trees in a graph is {\#}{P}-complete.
\newblock {\em Inf. Process. Lett.}, 51(3):111--116, 1994.

\bibitem{appr1}
Mark Jerrum and Alistair Sinclair.
\newblock Approximating the permanent.
\newblock {\em SIAM Journal on computing}, 18(6):1149--1178, 1989.

\bibitem{fkt2}
Pieter~W. Kasteleyn.
\newblock Graph theory and crystal physics.
\newblock In {\em Graph Theory and Theoretical Physics}, pages 43--110.
  Academic Press, 1967.

\bibitem{ladner}
Richard~E Ladner.
\newblock On the structure of polynomial time reducibility.
\newblock {\em Journal of the ACM (JACM)}, 22(1):155--171, 1975.

\bibitem{biclique}
Bingkai Lin.
\newblock The parameterized complexity of k-biclique.
\newblock In {\em Proceedings of the 26th Annual ACM-SIAM Symposium on Discrete
  Algorithms, SODA}, pages 605--615. Society for Industrial and Applied
  Mathematics, 2015.

\bibitem{lovasz}
L{\'a}szl{\'o} Lov{\'a}sz.
\newblock Operations with structures.
\newblock {\em Acta Mathematica Hungarica}, 18(3-4):321--328, 1967.

\bibitem{isoauto}
Rudolf Mathon.
\newblock A note on the graph isomorphism counting problem.
\newblock {\em Information Processing Letters}, 8(3):131--136, 1979.

\bibitem{nesetril}
Jaroslav Ne{\v{s}}et{\v{r}}il.
\newblock Homomorphisms of derivative graphs.
\newblock {\em Discrete Mathematics}, 1(3):257--268, 1971.

\bibitem{oxley}
James~G. Oxley.
\newblock {\em Matroid Theory}.
\newblock Oxford University Press, 1992.

\bibitem{feder}
Arash Rafiey, Jeff Kinne, and Feder Tom{\'{a}}s.
\newblock Dichotomy for digraph homomorphism problems.
\newblock {\em arXiv preprint arXiv:1701.02409}, 2017.

\bibitem{rota1964foundations}
Gian-Carlo Rota.
\newblock On the foundations of combinatorial theory {I}. {T}heory of
  m{\"o}bius functions.
\newblock {\em Probability theory and related fields}, 2(4):340--368, 1964.

\bibitem{schaefer}
Thomas~J Schaefer.
\newblock The complexity of satisfiability problems.
\newblock In {\em Proceedings of the tenth annual ACM Symposium on Theory of
  Computing, STOC}, pages 216--226. ACM, 1978.

\bibitem{stanley2011enumerative}
Richard~P Stanley.
\newblock Enumerative combinatorics: Volume 1.
\newblock 2011.

\bibitem{fkt1}
Harold N.~V. Temperley and Michael~E. Fisher.
\newblock Dimer problem in statistical mechanics - an exact result.
\newblock {\em Philosophical Magazine}, 6(68):1478--6435, 1961.

\bibitem{toda}
Seinosuke Toda.
\newblock {PP} is as hard as the polynomial-time hierarchy.
\newblock {\em SIAM Journal of Computing}, 20(5):865--877, 1991.

\bibitem{interpolation}
Salil~P Vadhan.
\newblock The complexity of counting in sparse, regular, and planar graphs.
\newblock {\em SIAM Journal on Computing}, 31(2):398--427, 2001.

\bibitem{Valiant1979a}
Leslie~G. Valiant.
\newblock The complexity of computing the permanent.
\newblock {\em Theoretical Computer Science}, 8(2):189--201, 1979.

\bibitem{accidentalalgos}
Leslie~G Valiant.
\newblock Accidental algorthims.
\newblock In {\em Proceedings of the 47th Annual Symposium on Foundations of
  Computer Science, FOCS}, pages 509--517. IEEE, 2006.

\bibitem{holographicalgos}
Leslie~G Valiant.
\newblock Holographic algorithms.
\newblock {\em SIAM Journal on Computing}, 37(5):1565--1594, 2008.

\bibitem{welsh}
Dominic~JA Welsh.
\newblock {\em Matroid theory}.
\newblock Courier Corporation, 2010.

\bibitem{rest1}
Mingji Xia, Peng Zhang, and Wenbo Zhao.
\newblock Computational complexity of counting problems on 3-regular planar
  graphs.
\newblock {\em Theoretical Computer Science}, 384(1):111--125, 2007.

\bibitem{zhuk}
Dmitriy Zhuk.
\newblock The proof of {CSP} dichotomy conjecture.
\newblock {\em arXiv preprint arXiv:1704.01914}, 2017.

\end{thebibliography}
\bibliographystyle{plain}

\end{document}